\documentclass[12pt, draftclsnofoot, journal, a4paper, onecolumn]{IEEEtran}
\usepackage[dvips]{color}
\usepackage{epsf}

\usepackage{amsmath}
\usepackage{amsxtra}
\usepackage{graphicx}
\usepackage{epsfig}
\usepackage{latexsym}
\usepackage{amsfonts}
\usepackage{rawfonts}
\usepackage[latin1]{inputenc}
\usepackage[T1]{fontenc}
\usepackage{calc}
\usepackage{capitalgreekitalic}
\usepackage{url}
\usepackage{enumerate}
\usepackage{color}
\usepackage{amssymb}
\usepackage{upref}
\usepackage{epic,eepic}
\usepackage{times}
\usepackage{cite}
\usepackage{dsfont}

\newcommand{\diag}{\ensuremath{\operatorname{diag}}}

\newtheorem{proposition}{{\bf Proposition}}
\newtheorem{lemma}{{\bf Lemma}}
\newtheorem{corollary}{{\bf Corollary}}

\newcommand{\qed}{\nobreak \ifvmode \relax \else
  \ifdim\lastskip<1.5em \hskip-\lastskip
  \hskip1.5em plus0em minus0.5em \fi \nobreak
  \vrule height0.75em width0.5em depth0.25em\fi}

\newcounter{step}
\newlength{\totlinewidth}
  {\end{list}%
  \rule{\linewidth}{1pt}}
\newcounter{substep}

  {\end{list}}

\newlength{\aligntop}
\newlength{\alignbot}
 \makeatletter
\renewenvironment{align}{%
  \vspace{\aligntop}
  \start@align\@ne\st@rredfalse\m@ne
}{%
  \math@cr \black@\totwidth@
  \egroup
  \ifingather@
    \restorealignstate@
    \egroup
    \nonumber
    \ifnum0=`{\fi\iffalse}\fi
  \else
    $$%
  \fi
  \ignorespacesafterend%
  \vspace{\alignbot}\par\noindent
} \makeatother

\IEEEoverridecommandlockouts

\begin{document}

\title{Opportunistic Relaying for Space-Time Coded Cooperation with Multiple Antenna Terminals}
\author{
\authorblockN{Behrouz Maham, \emph{Member, IEEE}, and Are Hj{\o}rungnes, \emph{Senior Member, IEEE}}\\
    \thanks{This work was supported by the Research Council of Norway
    through the project 176773/S10 entitled "Optimized Heterogeneous Multiuser MIMO Networks -- OptiMO". Behrouz Maham and
    Are Hj{\o}rungnes are with UNIK -- University Graduate Center, University of Oslo,
    Norway.
    This work was done during the stay of Behrouz Maham at Department of Electrical
    Engineering, Stanford University, USA. Emails: \protect\url{behrouzm@ifi.uio.no,arehj@unik.no}.}%
}
%
\maketitle

\begin{abstract}
We consider a wireless relay network with multiple antenna
terminals over Rayleigh fading channels, and apply distributed
space-time coding (DSTC) in \emph{amplify-and-forward} (A$\&$F)
mode.
The A$\&$F scheme is used in a way that each relay
transmits a scaled version of the linear combination of the received
symbols.
It turns out that, combined with power allocation in the relays,
A$\&$F DSTC results in an opportunistic relaying scheme, in which
only the \emph{best} relay is selected to retransmit the source's
space-time coded signal. Furthermore, assuming the knowledge of source-relay CSI at the source node, we design an efficient power allocation which outperforms uniform power allocation across the source antennas. Next, assuming $M$-PSK or $M$-QAM modulations, we analyze the performance of the proposed cooperative diversity transmission schemes
in a wireless relay networks with the
multiple-antenna source and destination. We derive the
probability density function (PDF) of the received SNR at the
destination. Then, the PDF is used to determine the symbol error
rate (SER) in Rayleigh fading channels. We derived closed-form
approximations of the average SER in the high SNR scenario, from which we find the
diversity order of system $R\min\{N_s, N_d\}$, where $R$, $N_s$, and $N_d$ are the number of the relays, source antennas, and destination antennas, respectively. Simulation results
show that the proposed system obtain more than 6 dB gain in SNR over A\&F MIMO DSTC for
BER $10^{-4}$, when $R=2$, $N_s=2$, and $N_d=1$.
\end{abstract}

\section{Introduction}
Space-time coding (STC) has received a lot of attention in the last
decade as a way of increasing the data rate and/or reduce the
transmitted power necessary to achieve a target bit error rate~(BER)
using multiple antenna transceivers. In ad-hoc network applications
or in distributed large scale wireless networks, the nodes are often
constrained in the complexity and size. This makes multiple-antenna
systems impractical for certain network applications~\cite{sen03a}.
In an effort to overcome this limitation, cooperative diversity
schemes have been introduced~\cite{sen03a,sen03b,lan00,lan02a}.
Cooperative diversity allows a collection of radios to relay signals
for each other and effectively create a virtual antenna array for
combating multipath fading in wireless channels. The attractive
feature of these techniques is that each node is equipped with only
\emph{one} antenna, creating a virtual antenna array. This property
makes them outstanding for deployment in cellular mobile devices as
well as in ad-hoc mobile networks, which have problems with
exploiting multiple-antenna due to the size limitation of the mobile
terminals.

Among the most widely used cooperative strategies are
amplify-and-forward (A\&F) \cite{lan02a}, \cite{nab04} and
decode-and-forward~(D\&F) \cite{sen03a,sen03b,lan02a}. The authors in
\cite{hua03} applied Hurwitz-Radon space-time codes in wireless
relay networks and conjecture a diversity factor around $R/2$ for
large $R$ from their simulations, where $R$ is the number of relays.

In \cite{jin06b}, a cooperative strategy was proposed, which
achieves a diversity factor of \emph{R} in a \emph{R}-relay wireless
network, using the so-called distributed space-time codes (DSTC).
In this strategy, a two-phase protocol is used. In the first phase, the
transmitter sends the information signal to the relays and in the second phase, the relays send information to the receiver. The signal sent by
every relay in the second phase is designed as a linear function of
its received signal. It was shown in \cite{jin06b} that the relays
can generate a linear space-time codeword at the receiver, as in a
multiple antenna system, although they only cooperate
distributively. This method does not require decoding at the relays
and for high SNR it achieves the optimal diversity
factor~\cite{jin06b}. Although distributed space-time coding does
not need instantaneous channel information in the relays, it
requires full channel information at the receiver of both the
channel from the transmitter to relays and the channel from relays
to the receiver. Therefore, training symbols have to be sent from
both the transmitter and the relays.
The design of practical A\&F
DSTCs that lead to reliable communication in wireless relay
networks, has also been recently considered~\cite{jin07,mah09twc,sus07}.

Distributed space-time coding was generalized to networks with
multiple-antenna nodes in \cite{jin05}. It was shown
that in a wireless network with $N_s$ antennas at the transmit node,
$N_d$ antennas at the receive node, and a total of $R$ antennas at
all relay nodes, the diversity order of $R\min\{N_s, N_d\}$ is
achievable \cite{jin05,pet08}.
In \cite{ogg08}, the problem of coding design considered over wireless relay network where both the transmitter and the receiver have several antennas.

Power efficiency is a critical design consideration for wireless
networks - such as ad-hoc and sensor networks - due to the limited
transmission power of the nodes. To that end, choosing the
appropriate relays to forward the source data, as well as the
transmit power levels of the source's antenna become important design
issues.
Several power allocation strategies for relay networks were studied
based on different cooperation strategies and network topologies in
\cite{Hon07}. In \cite{mah08v}, we proposed power allocation
strategies for repetition-based cooperation that take both the
statistical CSI and the residual energy information into account to
prolong the network lifetime while meeting the BER QoS requirement
of the destination. Distributed power allocation strategies for
D\&F cooperative systems were investigated in~\cite{che08}. Power allocation in three-node models are discussed in
\cite{hos05} and \cite{bro04}, while multi-hop relay networks are
studied in \cite{rez04,han04,doh04}.
The relay selection algorithms for networks with multiple relays can be also resulted
in power efficient transmission strategies. Recently proposed
practical relay selection strategies include pre-select one relay
\cite{luo05}, best-select relay \cite{luo05},
blind-selection-algorithm \cite{lin05}, informed-selection-algorithm
\cite{lin05}, and cooperative relay selection \cite{zhe05}. In
\cite{ble06b}, an opportunistic relaying scheme is introduced.
According to opportunistic relaying, a single relay among a set of
$R$ relay nodes is selected, depending on which relay provides for
the \emph{best} end-to-end path between source and destination.
Bletsas et al. \cite{ble06b} proposed two heuristic methods for
selecting the best relay based on the end-to-end instantaneous
wireless channel conditions. Performance and outage analysis of
these heuristic relay selection schemes were studied in \cite{ble07}
and \cite{zha06}.

In this paper, we propose decision metrics for
opportunistic relaying based on maximizing the received
instantaneous SNR at the destination in A\&F
mode, when both the source and destination have multiple-antennas.
We use a simple feedback from the destination toward
the relays to select the best relay and the best antenna at the source node.


%

%
Our main contributions can be summarized
as follows:
\begin{itemize}
  \item 
  We show that the distributed space-time codes (DSTC) based on
  \cite{jin06b} in a relay network with the multiple-antennas source and
  destination
  leads to a novel opportunistic relaying, when maximum instantaneous SNR based power allocation is employed.
  \item Assuming the knowledge of CSI of the source-relay links at the source, the optimum power allocations along the source's antennas based on maximizing the received SNR are derived.
  \item We analyze the performance of the proposed A$\&$F opportunistic relaying with
space-time coded source. In addition, the performance analysis of full-opportunistic scheme is studied, in which power control for both the source antennas and the relays are employed. More specifically, we derive the average
symbol error rate (SER) of opportunistic relaying and full-opportunistic schemes with $M$-PSK and
$M$-QAM modulations in a Rayleigh fading channels. Furthermore, the
probability density function (PDF) of the received SNR at the
destination is obtained.
\item For sufficiently
high SNR, simple closed-form average SER expressions are derived for
A$\&$F opportunistic relaying links with multiple cooperating branches
and multiple antennas source/destination. Based on the proposed
approximated SER expression, it is shown that the proposed schemes
achieve the diversity order of $R\min\{N_s, N_d\}$, where $R$,
$N_s$, and $N_d$ are the number of relays, source antennas, and the
destination antennas, respectively.
\item We verify the obtained analytical results using
simulations. The results show that the derived error rates have the same
system performance as simulation results. Assuming $R=2$, $N_s=2$,
$N_d=1$, the proposed opportunistic scheme outperforms DSTC by about
6 dB gain in SNR at BER $10^{-4}$.
\end{itemize}


The remainder of this paper is organized as follows:
In Section II, the system model is given.
The power control strategies for A\&F DSTC based on the availability of CSI at the source and relays are
considered in Section III.
The average SER of the proposed opportunistic schemes under $M$-PSK and $M$-QAM modulations are derived in Section IV. In Section V, closed-form approximations for the average SER are presented, and the diversity analysis is carried out.
In Section VI, the overall system performance is presented via simulations for
different numbers of relays, source and destination antennas, and the correctness of the analytical
formulas are confirmed by Monte Carlo simulations.
Conclusions are presented in Section VII. The article contains four
appendices which present various proofs.

\emph{Notations}: The superscripts $^t$ and $^H$ stand for
transposition and conjugate transposition, respectively.
The expectation value operation is denoted by $\mathbb{E}\{\cdot\}$.
%
%
The symbol $\boldsymbol{I}_T$ stands for the $T\times T$ identity
matrix. $\|\boldsymbol{A}\|$ denotes the Frobenius norm of the
matrix $\boldsymbol{A}$. The trace of the matrix $\boldsymbol{A}$ is
denoted by $\text{tr}\left\{\boldsymbol{A}\right\}$.
$\diag\{\boldsymbol{A}_1,\ldots,\boldsymbol{A}_R\}$ denotes the
block diagonal matrix.

\section{System Model}

\begin{figure}[e]
  \centering
  \vspace{-2.5cm}\hspace*{-.5cm}
  \includegraphics[width=\columnwidth]{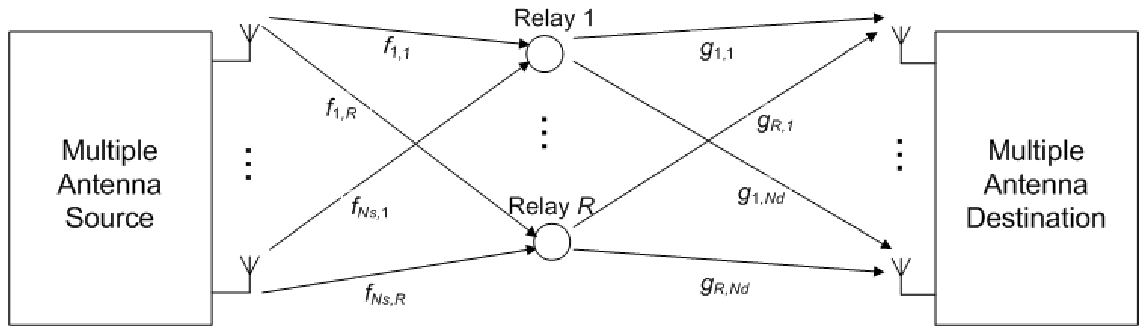}\\
  \vspace{-2.5cm}
  \caption{Wireless relay network including one source with $N_s$ antennas, $R$
relays, and one destination with $N_d$ antennas.
  }\label{f1}
\end{figure}

Consider a wireless communication scenario where the source node
\emph{s} transmits information to the destination node \emph{d} with
the assistance of one or more relays denoted Relay $r=1, 2, \ldots,
R$ (see Fig. \ref{f1}). The source and destination nodes are
equipped with $N_s$ and $N_d$ antennas, respectively. Without loss
of generality, it is assumed that each relay nodes is equipped with
a single antennas. Note that this network can be transformed to
relays with multiple antenna, since the transmit and receive signals
at different antennas of the same relay can be processed and
designed independently.

We denote the links from the $N_s$ source antennas to the $r$th
relay as $f_{1,r},f_{2,r},\ldots,f_{N_s,r}$, and the links from the
$r$th relay to the $N_d$ destination antennas as
$g_{r,1},g_{r,2},\ldots,g_{r,N_d}$.
Under the assumption that each link undergoes independent Rayleigh
process $f_{i,r}$, and $g_{r,j}$ are independent complex Gaussian
random variables with zero-mean and variances $\sigma_{f_{r}}^2$,
and $\sigma_{g_r}^2$, respectively. Since the multiple antennas in
source and destination are co-located, and the co-located antennas
have the same distances to relays, we skipped the $i$ and $j$
indices of $\sigma_{f_{r}}^2$ and $\sigma_{g_r}^2$.

Assume that the source wants to send $K$ symbols $s_1$, $s_2$,
$\ldots$, $s_K$ to the destination during $T$ time slots. $T$ should
be less than the coherent interval, that is, the time duration among
which the channels $f_{i,r}$, and $g_{r,j}$ are constant. Henceforth,
we assume using full-rate space-time codes, and thus, $K=T$. Similar
to \cite{jin06b}, our scheme requires two phases of transmission.
During the first phase, the source should transmit a $T\times N_s$
dimensional orthogonal code matrix $\boldsymbol{S}_1$ to \emph{all}
relays. We can represent $\boldsymbol{S}_1$ in terms of the vector
$\boldsymbol{s}=[s_1, s_2, \ldots, s_T]^t$ as
\begin{equation}\label{1}
    \boldsymbol{S}_1=[\boldsymbol{A}_1\boldsymbol{s} \,\boldsymbol{A}_2 \boldsymbol{s} \,\ldots\, \boldsymbol{A}_{N_s}\boldsymbol{s}],
\end{equation}
where $\boldsymbol{A}_i$, $i=1,\ldots,N_s$, are $T\times T$ unitary
matrices, and $\boldsymbol{s}_i=\boldsymbol{A}_i\boldsymbol{s}$
describes the $i$th column of a $T\times N_s$ orthogonal space-time
code. We assume the following normalization
\begin{equation}\label{2}
    \mathbb{E}\left[\,\text{tr}\{\boldsymbol{S}_1^H\boldsymbol{S}_1\}\right]=\mathbb{E}\left[\,\text{tr}\left\{\sum_{k=1}^T|s_k|^2\boldsymbol{I}_{N_s}\right\}\right]=N_s.
\end{equation}
The source transmits $\sqrt{P_1T/N_s}\boldsymbol{S}_1$ where $P_1T$
is the average total power used at the source during the first
phase. Thus, $\sqrt{P_1T/N_s}\boldsymbol{s}_i$, $i=1,\ldots,N_s$, is
the signal sent by the $i$th antenna with the average power of
$P_1T/N_s$. Assuming that $f_{i,r}$ does not vary during $T$
successive intervals, the $T\times1$ receive signal vector at the
$r$th relay is
%
%
\begin{equation}\label{3}
    \boldsymbol{x}_r=\sqrt{\frac{P_1T}{N_s}}\boldsymbol{S}_1\boldsymbol{f}_r+\boldsymbol{v}_r,
\end{equation}
where
$\boldsymbol{f}_r=[{f}_{1,r}\,{f}_{2,r}\,\ldots\,{f}_{N_s,r}]^t$,
and $\boldsymbol{v}_r$ is a $T\times1$ complex zero-mean white
Gaussian noise vector with
variance $\mathcal{N}_1$.

In the second phase of the transmission, all relays simultaneously
transmit linear functions of their received signals
$\boldsymbol{x}_r$. In order to construct a distributed space-time
codes, the received signal at the $j$th antenna of the
\emph{destination} is collected inside the $T\times1$ vector
$\boldsymbol{y}_j$ as
%
%
\begin{equation}\label{4}
    \boldsymbol{y}_j=\sum_{r=1}^Rg_{r,j}\,\rho_r\boldsymbol{C}_r\boldsymbol{x}_r+\boldsymbol{w}_j,
\end{equation}
for $j=1,2,\ldots,N_r$, where $\boldsymbol{w}_j$ is a $T\times1$
complex zero-mean white Gaussian noise vector with component-wise
variance $\mathcal{N}_2$, $\rho_r$ is the scaling factor at Relay
$r$, and $\boldsymbol{C}_r$, of size $T\times T$, are obtained by
representing the $r$th column of an appropriate $T\times R$
dimensional space-time code matrix as
$\boldsymbol{C}_r\boldsymbol{s}$. This construction method
originates from the construction of a space-time code for co-located
multiple-antenna systems, where the transmitted signal vector from
the $k$th antenna is $\boldsymbol{C}_k\boldsymbol{s}$ \cite{wan03}.
When there is no instantaneous channel state information (CSI) at
the relays, but statistical CSI is known, a useful constraint is to
ensure that a given average transmitted power is maintained. That
is,
\begin{equation}\label{5}
    \rho_r=\sqrt{\frac{P_{2,r}}{\sigma_{f_r}^2P_1+\mathcal{N}_1}},
\end{equation}
where $P_{2,r}$ is the average transmitted power from Relay
\emph{r}.

We can further represent input-output relationship of the DSTC as
the space-time code in a multiple-antenna system.
By setting the $T\times N_s R$ space-time encoded signal
\begin{equation}\label{6a}
\boldsymbol{S}=[\boldsymbol{C}_{1}\boldsymbol{S}_1,\,\boldsymbol{C}_{2}\boldsymbol{S}_1,\,\ldots,\,\boldsymbol{C}_{R}\boldsymbol{S}_1],
\end{equation}
and by concatenating the received signals of the destination
antennas, i.e.,
$\boldsymbol{Y}=[\boldsymbol{y}_{1}\,\boldsymbol{y}_{2}\,\ldots\,\boldsymbol{y}_{N_d}]\in\mathds{C}^{T\times N_d}$,
from \eqref{3}-\eqref{4}, we have
\begin{equation}\label{6}
    \boldsymbol{Y}=\sqrt{\frac{P_1T}{N_s}}\boldsymbol{S}\boldsymbol{H}+\boldsymbol{W}_T.
\end{equation}
The $N_sR\times N_d$ channel matrix $\boldsymbol{H}$ in
\eqref{6} can be written as
\begin{equation}\label{6z}
    \boldsymbol{H}=\boldsymbol{F}\boldsymbol{\Lambda}\boldsymbol{G},
\end{equation}
where matrices $\boldsymbol{F}$, $\boldsymbol{\Lambda}$, and $\boldsymbol{G}$ of sizes $N_s R\times R$, $R\times R$, $R\times N_d$, respectively, are given by
\begin{equation*}
    \boldsymbol{F}=\diag\left\{
                      \boldsymbol{f}_1,
                      \ldots,
                      \boldsymbol{f}_R
\right\},\,\,\,\,\,\,\,\,\,     \boldsymbol{\Lambda}=\diag\left\{
                      \rho_1,
                      \ldots,
                      \rho_R
\right\},
\end{equation*}
\begin{equation*}
    \boldsymbol{g}_r=[{g}_{r,1}\,{g}_{r,2}\,\ldots\,{g}_{r,N_d}],\,\,\,\,\,\,\,\,\,     \boldsymbol{G}=\left[
                      \boldsymbol{g}_1^t,
                      \ldots,
                      \boldsymbol{g}_R^t \right]^t.
\end{equation*}
The total noise in \eqref{6} is collected into the $T\times N_d$
matrix
\begin{equation}\label{7}
    \boldsymbol{W}_T=\boldsymbol{V}\boldsymbol{\Lambda}\boldsymbol{G}+\boldsymbol{W},
\end{equation}
where
$\boldsymbol{V}=\left[\boldsymbol{C}_{1}\boldsymbol{v}_1\,\boldsymbol{C}_{2}\boldsymbol{v}_2\,\ldots\,
\boldsymbol{C}_{R}\boldsymbol{v}_{R}\,\right]\in\mathds{C}^{T\times R}$
and
$\boldsymbol{W}=\left[\boldsymbol{w}_1\,\boldsymbol{w}_2\,\ldots\,\boldsymbol{w}_{N_d}\,\right]\in\mathds{C}^{T\times N_d}$.

Since in this paper, we focus on orthogonal design, the maximum
likelihood (ML) detection is decomposed to single-symbol detection,
maximal-ratio combining (MRC) can be applied at the destination \cite{sim00}. To calculate the post detection SNR at the output of the ML DSTC decoder, we need to compute the received signal power. Hence, using \eqref{6}, we have
\begin{align}\label{10}
    \eta_{{s}_d}&=\frac{P_1T}{N_s}\mathbb{E}_s\!\left[\text{tr}\{\boldsymbol{S}\boldsymbol{H}\boldsymbol{H}^H\boldsymbol{S}^H\}\right]
    \!=\!\frac{P_1T}{N_s}\mathbb{E}_s\!\left[\text{tr}\{\boldsymbol{H}\boldsymbol{H}^H\boldsymbol{S}^H\boldsymbol{S}\}\right]
    =\frac{P_1T}{N_s}\,\text{tr}\{\boldsymbol{H}\boldsymbol{H}^H\,\mathbb{E}_s[\boldsymbol{S}^H\boldsymbol{S}]\}.
\end{align}

To have the linear orthogonal ML detection, we should design the
DSTC, such that
\begin{align}\label{11}
\boldsymbol{S}^H\boldsymbol{S}=(|s_1|^2+|s_2|^2+\ldots+|s_T|^2)\boldsymbol{I}_{N_sR},
\end{align}
and using the normalization assumed in \eqref{2}, we have
$\mathbb{E}_s[\boldsymbol{S}^H\boldsymbol{S}]=\boldsymbol{I}_{N_sR}$. For designing the distributed orthogonal space-time codes in multiple-antenna relay netowrks, one can see \cite{mah10}.
Thus, $\eta_{{s}_d}$ in \eqref{10} can be evaluated as
\begin{align}\label{13}
    \eta_{{s}_d}&=\frac{P_1T}{N_s}\,\text{tr}\{\boldsymbol{H}\boldsymbol{H}^H\}=\frac{P_1T}{N_s}\,\sum_{i=1}^{N_sR}\left[\boldsymbol{H}\boldsymbol{H}^H\right]_{i,i}
\nonumber\\
&=\frac{P_1T}{N_s}\!\sum_{r=1}^{R}\!\sum_{n=1}^{N_s}|f_{n,r}|^2\rho_r^2\sum_{j=1}^{N_d}|g_{r,j}|^2
=\frac{P_1T}{N_s}\!\sum_{r=1}^{R}\rho_r^2\|\boldsymbol{f}_{r}\|^2\|\boldsymbol{g}_{r}\|^2.
\end{align}

From \eqref{7}, and assuming $\boldsymbol{C}_r$, $r=1,\ldots,R$ are unitary matrices, the total noise power at the destination can be
written as
\begin{equation}\label{14}
    \eta_{{w}_T}=\mathbb{E}_{v,w}[\text{tr}\{\boldsymbol{W}_T\boldsymbol{W}_T^H\}]=
    T\left(\sum_{k=1}^R\rho_k^2\|\boldsymbol{g}_{k}\|^2\mathcal{N}_1+N_d\mathcal{N}_2\right).
\end{equation}

Combining \eqref{13} and \eqref{14}, the received SNR at the
destination can be written as
\begin{align}\label{15}
    \text{SNR}_{\text{ins}}=\sum_{r=1}^{R}\frac{P_1\rho_r^2\|\boldsymbol{f}_{r}\|^2\|\boldsymbol{g}_{r}\|^2}{N_s\displaystyle\sum_{k=1}^R\rho_k^2\|\boldsymbol{g}_{k}\|^2\mathcal{N}_1+N_sN_d\mathcal{N}_2}.
\end{align}

\section{Power Control in A$\&$F Space-Time Coded Cooperation}
In this section, we propose power allocation schemes for the A$\&$F
distributed space-time codes with multiple antennas
source/destination, based on maximizing the received SNR at the
destination \emph{d}. First, we will find the optimum distribution
of transmitted powers among relays, i.e., $P_{2,r}$, based on
instantaneous SNR. Then, the optimum power transmitted in the two
phases, i.e., $P_1$ and $P_2=\sum_{r=1}^R P_{2,r}$, will be obtained
by maximizing the average received SNR at the destination.

\subsection{Power Control among Relays with No CSI at the Source}
Here, we find the optimum distribution of the transmitted powers
among relays during the second phase, in a sense of maximizing the
instantaneous SNR at the destination.

\subsubsection{Optimum Power Allocation}
Using \eqref{5} and \eqref{15}, the instantaneous received SNR at the destination can be
written as
\begin{equation}\label{11b}
\text{SNR}_{\text{ins}}=\frac{\boldsymbol{p}^t\boldsymbol{U}\boldsymbol{p}}
{\boldsymbol{p}^t\boldsymbol{Q}\boldsymbol{p}+N_s N_d\mathcal{N}_2},
\end{equation}
where
$\boldsymbol{p}=[\sqrt{P_{2,1}},\sqrt{P_{2,2}},\ldots,\sqrt{P_{2,R}}]^t$
and diagonal $R\times R$ matrices $\boldsymbol{U}$ and $\boldsymbol{V}$ are
defined as
\begin{align}\label{12b}
\boldsymbol{U}&\!\!=\!\diag\!\!\left[\!\frac{P_{\!1}\|\boldsymbol{f}_{\!1}\|^2\|\boldsymbol{g}_{\!1}\|^2}
{\sigma_{f_1}^2P_1+\mathcal{N}_1},
\frac{P_{\!1}\|\boldsymbol{f}_{\!2}\|^2\|\boldsymbol{g}_{\!2}\|^2}
{\sigma_{f_2}^2P_1+\mathcal{N}_1},\ldots,\frac{P_{\!1}\|\boldsymbol{f}_{\!\!R}\|^2\|\boldsymbol{g}_{\!R}\|^2}
{\sigma_{f_R}^2P_1+\mathcal{N}_1}\!\right]\!,
\nonumber\\
\boldsymbol{Q}&=
\diag\left[\frac{N_s\|\boldsymbol{g}_{1}\|^2\mathcal{N}_1}
{\sigma_{f_1}^2P_1+\mathcal{N}_1},
\frac{N_s\|\boldsymbol{g}_{2}\|^2\mathcal{N}_1}
{\sigma_{f_2}^2P_1+\mathcal{N}_1},\ldots,
\frac{N_s\|\boldsymbol{g}_{R}\|^2\mathcal{N}_1}
{\sigma_{f_R}^2P_1+\mathcal{N}_1}\right].
\end{align}

Then, the optimization problem is formulated as
\begin{equation}\label{13b}
\boldsymbol{p}^*=\arg\max_{\boldsymbol{p}}\text{SNR}_{\text{ins}},\,\,\,\,\text{subject
to}\,\,\,\, \boldsymbol{p}^t\boldsymbol{p}=P_2,
\end{equation}
where the $R\times 1$ vector $\boldsymbol{p}^*$ denotes the optimum
values of power control coefficients. Since
$\boldsymbol{p}^t\boldsymbol{p}=P_2$, we can rewrite \eqref{11b} as
$\text{SNR}_{\text{ins}}=\frac{\boldsymbol{p}^t\boldsymbol{U}\boldsymbol{p}}
{\boldsymbol{p}^t\boldsymbol{W}\boldsymbol{p}}$,
where diagonal matrix $\boldsymbol{W}$ is defined as
$\boldsymbol{W}=\boldsymbol{Q}+\frac{N_s
N_d\mathcal{N}_2}{P_2}\boldsymbol{I}_T$. Since $\boldsymbol{W}$ is a
real-valued positive semi-definite matrix, we define
$\boldsymbol{q}\triangleq\boldsymbol{W}^{\frac{1}{2}}\boldsymbol{p}$,
where
$\boldsymbol{W}=(\boldsymbol{W}^{\frac{1}{2}})^t\boldsymbol{W}^{\frac{1}{2}}$.
Then, $\text{SNR}_{\text{ins}}$ can be rewritten as
\begin{equation}\label{15b}
\text{SNR}_{\text{ins}}=\frac{\boldsymbol{q}^t\boldsymbol{Z}\boldsymbol{q}}
{\boldsymbol{q}^t\boldsymbol{q}},
\end{equation}
where diagonal matrix $\boldsymbol{Z}$ is
$\boldsymbol{Z}=\boldsymbol{U}\boldsymbol{W}^{-1}$. Now, using
Rayleigh-Ritz theorem \cite{hor85}, we have 
\begin{equation}\label{16b}
\frac{\boldsymbol{q}^t\boldsymbol{Z}\boldsymbol{q}}
{\boldsymbol{q}^t\boldsymbol{q}}\leq\lambda_{\max},
\end{equation}
where $\lambda_{\max}$ is the largest eigenvalue of
$\boldsymbol{Z}$, which is corresponding to the largest diagonal
element of $\boldsymbol{Z}$, i.e.,
\begin{align}\label{17b}
\lambda_{\max}\!&=\!\max_{i\in\{1,\ldots,R\}}\!\!\lambda_i
=
\!\max_{i\in\{1,\ldots,R\}}\frac{P_1P_2\|\boldsymbol{f}_{i}\|^2\|\boldsymbol{g}_{i}\|^2}{P_2N_s
\|\boldsymbol{g}_{i}\|^2\!\mathcal{N}_1\!+\!N_s
N_d\mathcal{N}_2(\sigma_{f_i}^2P_1\!+\!\mathcal{N}_1\!)}.
\end{align}
The equality in $\frac{\boldsymbol{q}^t\boldsymbol{Z}\boldsymbol{q}}
{\boldsymbol{q}^t\boldsymbol{q}}=\lambda_{\max}$ holds if
$\boldsymbol{q}$ is proportional to the eigenvector of
$\boldsymbol{Z}$ corresponding to $\lambda_{\max}$. Using the
eigenvalue decomposition of the diagonal matrix $\boldsymbol{Z}$, which contains positive diagonal elements,
it is obvious that the matrix which is
consisting of the normalized eigenvectors, is the identity matrix.
Hence, the optimum $\boldsymbol{q}_{\max}$ is proportional to
$\boldsymbol{e}_{i_{\max}}$, which is a $R\times 1$ vector with only
zero elements, except one at the $i_{\max}$-th component. On the
other hand, since
$\boldsymbol{p}=\boldsymbol{W}^{-\frac{1}{2}}\boldsymbol{q}$, and
$\boldsymbol{W}$ is a diagonal matrix, the optimum
$\boldsymbol{p}^*$ is also proportional to
$\boldsymbol{e}_{i_{\max}}$. Using the power constraint of the
transmitted power in the second phase, i.e.,
$\boldsymbol{p}^t\boldsymbol{p}=P_2$, we have
$\boldsymbol{p}^*=\sqrt{P_2}\boldsymbol{e}_{i_{\max}}$. This means
that for each realization of the network channels, the best relay
should transmit all the available power $P_2$, while all the other relays
should stay silent. 



\subsubsection{Relay Selection Strategy}
The process of selecting the best relay could be done by the
destination. This is feasible since the destination node should be
aware of both the backward and forward channels for coherent decoding.
Thus, the same channel information could be exploited for the
purpose of relay selection. However, if we assume a distributed
relay selection algorithm, in which relays independently decide to
select the best relay among them, such as work done in
\cite{ble06b}, the knowledge of local channels $f_i$ and $g_i$ is
required for the $i$th relay. The estimation of $f_i$ and $g_i$ can
be done by transmitting a ready-to-send (RTS) packet and a
clear-to-send (CTS) packet in MAC protocols.

\subsection{Power Allocation with Partial CSI at the Source}
Here, we study the situation in which the CSI of the source-relay
links are known at the source node. In this case, instead of uniform
power allocation used in the previous subsection, power allocation
is used over the transmit antennas. Thus, \eqref{3} can be rewritten
as
\begin{equation}\label{3q}
    \boldsymbol{x}_r=\sqrt{T}\boldsymbol{X}_{\!p}\boldsymbol{f}_r+\boldsymbol{v}_r,
\end{equation}
where $\boldsymbol{X}_{\!p}=[P_{1,1}\boldsymbol{A}_1\boldsymbol{s}
\,P_{1,2}\boldsymbol{A}_2 \boldsymbol{s} \,\ldots\,
P_{1,N_s}\boldsymbol{A}_{N_s}\boldsymbol{s}]$,
$\sum_{k=1}^{N_s}P_{1,k}=P_1$, and $P_{1,k}$, $k=1,\ldots,N_s$, is
the transmit power from the $k$th source antenna.

Hence, using \eqref{4}-\eqref{6} and \eqref{3q}, $ \eta_{{s}_d}$ in \eqref{10} can be rewritten as
\begin{align}\label{10q}
    \eta_{{s}_d}&=T\mathbb{E}_s\!\left[\text{tr}\{\boldsymbol{S}_p\boldsymbol{H}\boldsymbol{H}^H\boldsymbol{S}_p^H\}\right]
    \!=\!T\mathbb{E}_s\!\left[\text{tr}\{\boldsymbol{H}\boldsymbol{H}^H\boldsymbol{S}_p^H\boldsymbol{S}_p\}\right]
=T\,\text{tr}\{\boldsymbol{H}\boldsymbol{H}^H\,\mathbb{E}_s[\boldsymbol{S}_p^H\boldsymbol{S}_p]\},
\end{align}
where
\begin{align}
\boldsymbol{S}_p=[&P_{1,1}\boldsymbol{C}_{1}\boldsymbol{A}_1\boldsymbol{s},\,
\ldots,P_{1,N_s}\boldsymbol{C}_{1}\boldsymbol{A}_{N_s}\boldsymbol{s},\,\ldots,
P_{1,1}\boldsymbol{C}_{R}\boldsymbol{A}_1\boldsymbol{s}
,\ldots,P_{1,N_s}\boldsymbol{C}_{R}\boldsymbol{A}_{N_s}\boldsymbol{s}],
\end{align}
has the size $T\times N_sR$ and using the normalization assumed in \eqref{2}, we have
\begin{align}\label{12q}
\mathbb{E}_s[\boldsymbol{S}_p^H\boldsymbol{S}_p]=\diag(P_{1,1},\ldots,P_{1,N_s})\otimes\boldsymbol{I}_{R},
\end{align}
where $\otimes$ is Kroncker product. Thus, $\eta_{{s}_d}$ can be evaluated as
\begin{align}\label{13q}
    \eta_{{s}_d}&=T\sum_{r=0}^{R-1}\sum_{n=1}^{N_s}P_{1,n}\left[\boldsymbol{H}\boldsymbol{H}^H\right]_{N_s r+n,N_s r+n}
=T\!\sum_{r=1}^{R}\!\sum_{n=1}^{N_s}P_{1,n}|f_{n,r}|^2\rho_r^2\|\boldsymbol{g}_{r}\|^2.
\end{align}

\subsubsection{Optimum Power Allocation}
For deriving the optimum value of power in a sense of minimizing the received SNR, we have to compute $\text{SNR}_{\text{ins}}=\frac{\eta_{{s}_d}}{\eta_{w}}$.
Combining \eqref{14} and \eqref{13q}, the received SNR at the
destination can be written as $\text{SNR}_{\text{ins}}=\sum_{n=1}^{N_s}\theta_n P_{1,n}$ where
\begin{align}\label{16s}
    \theta_n=\frac{\sum_{r=1}^{R}|f_{n,r}|^2\rho_r^2\|\boldsymbol{g}_{r}\|^2}{N_d\sum_{k=1}^R\rho_k\|
    \boldsymbol{g}_{k}\|^2\mathcal{N}_1+N_d\mathcal{N}_2}.
\end{align}
Hence, we can formulate the following problem to find the optimum values of
$P_{1,n}$:
\begin{align}\label{12g}
    &\max\,\, \sum_{n=1}^{N_s}\theta_n P_{1,n},
    \nonumber\\
    &\text{s.t.}\,\,\,
    \sum_{n=1}^{N_s} P_{1,n}\leq P_{1},
    \nonumber\\
    &\,\,\,\,\,\,\,\,\,\,\,0\leq P_{1,n},\,\,\text{for}\,\,n=1,\ldots,
    N.
\end{align}

The optimization problem in \eqref{12g} is a maximal assignment
problem, and it is easy to show that the solution to this problem is
\begin{align}\label{e}
    P_{1,n}^*=\left\{\begin{array}{cc}
                                                P_1, & \text{if}\,\,n=\arg \displaystyle\max_{i\in \{1,\ldots,N_s\}}
                                                \theta_i,
\\
0, & \text{otherwise}. \\
                                              \end{array}
    \right.
\end{align}
Therefore, the optimum solution for the problem stated in \eqref{12g} is such that the whole power in the first phase is transmitted by an antenna at the source with the highest value of $\theta_n$ in \eqref{16s}.
From \eqref{e}, we can rewrite the received SNR as
$\text{SNR}_{\text{ins}}=\sum_{r=1}^{R}\zeta_r$ where
\begin{align}\label{24kkk}
\zeta_r=\frac{P_1P_{2,r}\displaystyle\max_{n\in\{1,\ldots,N_s\}}|f_{n,r}|^2\|\boldsymbol{g}_{r}\|^2}{P_2
\|\boldsymbol{g}_{r}\|^2\!\mathcal{N}_1\!+\!
N_d\mathcal{N}_2(\sigma_{\xi_r}^2P_1\!+\!\mathcal{N}_1\!)}.
\end{align}

Now, by defining $\xi_r=\max_{n\in\{1,\ldots,N_s\}}|f_{n,r}|^2$ with mean $\sigma_{\xi_r}^2$, we can employ a similar procedure used in the previous subsection to find the optimal values of $P_{2,r}$, and the matrices $\boldsymbol{U}$ and $\boldsymbol{Q}$ in \eqref{12b} are redefined as
\begin{align}\label{12bb}
\boldsymbol{U}&\!\!=\!\diag\!\!\left[\!\frac{P_{\!1}\xi_1\|\boldsymbol{g}_{\!1}\|^2}
{\sigma_{\xi_1}^2P_1+\mathcal{N}_1},
\frac{P_{\!1}\xi_2\|\boldsymbol{g}_{2}\|^2}
{\sigma_{\xi_2}^2P_1+\mathcal{N}_1},\ldots,\frac{P_{\!1}\xi_{\!R}\|\boldsymbol{g}_{\!R}\|^2}
{\sigma_{\xi_R}^2P_1+\mathcal{N}_1}\!\right]\!,
\nonumber\\
\boldsymbol{Q}&=
\diag\left[\frac{N_s\|\boldsymbol{g}_{1}\|^2\mathcal{N}_1}
{\sigma_{\xi_1}^2P_1+\mathcal{N}_1},
\frac{N_s\|\boldsymbol{g}_{2}\|^2\mathcal{N}_1}
{\sigma_{\xi_2}^2P_1+\mathcal{N}_1},\ldots,
\frac{N_s\|\boldsymbol{g}_{R}\|^2\mathcal{N}_1}
{\sigma_{\xi_R}^2P_1+\mathcal{N}_1}\right].
\end{align}

Therefore, similar to \eqref{17b}, the whole transmission power should be sent from a best relay in the optimal setting. The relay with highest value of $\frac{P_1\displaystyle\xi_j\|\boldsymbol{g}_{j}\|^2}{P_2
\|\boldsymbol{g}_{j}\|^2\!\mathcal{N}_1\!+\!
N_d\mathcal{N}_2(\sigma_{\xi_j}^2\!P_1\!+\!\mathcal{N}_1\!)}
$ is selected as the best relay, and its corresponding power is chosen as
\begin{align}\label{e2}
    P_{2,r}^*=\left\{\begin{array}{cc}
                                                P_2, & \text{if}\,\,r=\arg \!\!\displaystyle\max_{j\in \{1,\ldots,R\}}\!
                                                \frac{P_1\displaystyle\xi_j\|\boldsymbol{g}_{j}\|^2}{P_2
\|\boldsymbol{g}_{j}\|^2\!\mathcal{N}_1\!+\!
N_d\mathcal{N}_2(\sigma_{\xi_j}^2\!P_1\!+\!\mathcal{N}_1\!)},
\\
0, & \text{otherwise}. \\
                                              \end{array}
    \right.
\end{align}

\subsubsection{Relay and Source Antenna Selection Strategy}
Based on \eqref{e} and \eqref{e2}, we can summarize the process of relay and source's antenna selection as follows:\\
(1) Choose the best relay such that $r^*=\arg \displaystyle\max_{r\in\{1,\ldots,R\}} \frac{\displaystyle\max_{n\in\{1,\ldots,N_s\}}|f_{n,r}|^2\|\boldsymbol{g}_{r}\|^2}{P_2
\|\boldsymbol{g}_{r}\|^2\!\mathcal{N}_1\!+\!
N_d\mathcal{N}_2(\sigma_{\xi_r}^2\!P_1\!+\!\mathcal{N}_1\!)}$.\\
(2) After finding $r^*$, choose the $n^*$th antenna at the source as the best antenna, such that $n^*=\arg \displaystyle\max_{n\in\{1,\ldots,N_s\}} |f_{n,r^*}|^2$.

\subsection{Power Allocation with CSI at the Source and No CSI at Relays}
Here, we study the situation in which the CSI of the source-relay
links are known at the source node, when no power allocation is used at the relays. In this case, we employ DSTC with uniform power allocation at the relays. 

From \eqref{16s} and by assuming the equal power allocation among relays is used, i.e., $P_{2,r}=P_2$, $r=1,\ldots,R$, we have $\text{SNR}_{\text{ins}}=\sum_{n=1}^{N_s}P_{1,n}\theta_n$ where $\theta_n$ can be rewritten as
\begin{align}\label{16ss}
    \theta_n=\frac{\sum_{r=1}^{R}|f_{n,r}|^2\frac{P_2}{\sigma_{f_{n^{\!*}\!,r}}^2\!P_1\!+\!\mathcal{N}_1}
    \|\boldsymbol{g}_{r}\|^2}{N_d\sum_{k=1}^R\frac{P_2}{\sigma_{f_{n^{\!*}\!,k}}^2\!P_1\!+\!\mathcal{N}_1}\|
    \boldsymbol{g}_{k}\|^2\mathcal{N}_1+N_d\mathcal{N}_2},
\end{align}
where $\sigma_{f_{n^{\!*}\!,r}}^2$ is the mean of the random variable $|f_{n^*,r}|^2$, and $n^*$ denotes the index of the selected antenna at the source.

Similar to the optimization problem stated in \eqref{12g}, we can find the optimal value of $P_{1,n}$ from \eqref{e}. Moreover, by defining $\tilde{\theta}_n$ as
\begin{align}\label{16s2}
    \tilde{\theta}_n=\sum_{r=1}^{R}\frac{|f_{n,r}|^2\|\boldsymbol{g}_{r}\|^2}{\sigma_{f_{n^{\!*}\!,r}}^2\!P_1\!+\!\mathcal{N}_1},
\end{align}
we can equivalently find the optimal value of
$P_{1,n}$ given by
\begin{align}\label{e3}
    P_{1,n}^*=\left\{\begin{array}{cc}
                                                P_1, & \text{if}\,\,n=\arg \displaystyle\max_{i\in \{1,\ldots,N_s\}}
                                                \tilde{\theta}_i
\\
0, & \text{otherwise} \\
                                              \end{array}
    \right.
\end{align}
Therefore, the whole power in the first phase is transmitted by the antenna at the source with the highest value of $\tilde{\theta}_n=\sum_{r=1}^{R}\frac{|f_{n,r}|^2\|\boldsymbol{g}_{r}\|^2}{\sigma_{f_{n^{\!*}\!,r}}^2\!P_1\!+\!\mathcal{N}_1}$. The transmission power $P_2$ in the second phase can be chosen equally as $\frac{P}{2R}$.

Note that the process selection of the best antenna at the source can be done at the destination in which we have access to the CSI. Then, the index of the selected antenna at the source is fed back to the source.

\subsection{Power Control between Two Phases}
In the following proposition, we derive the optimal value for the
transmitted power in the two phases when backward and forward
channels have different variances by maximizing the average SNR at
the destination. 

\begin{proposition}\label{a}
Assume $\tau$ portion of the total power is transmitted in the
first phase and the remaining power is transmitted by relays at the
second phase, where $0<\tau<1$, that is $P_1=\tau P$ and
$P_2=(1-\tau) P$, where $P$ is the total transmitted power during
two phases. Assuming $\sigma_{f_r}^2=\sigma_{f}^2$ and
$\sigma_{g_r}^2=\sigma_{g}^2$, the optimum value of $\tau$ by
maximizing the average SNR at the destination is
\begin{equation}\label{5x}
\tau=\frac{\mathcal{N}_1\sigma_g^2P+\mathcal{N}_1\mathcal{N}_2}
{(\mathcal{N}_2\sigma_f^2-\mathcal{N}_1\sigma_g^2)\,P}\left(\sqrt{1+\frac{(\mathcal{N}_2\sigma_f^2
-\mathcal{N}_1\sigma_g^2)\,P}{\mathcal{N}_1\sigma_g^2P+\mathcal{N}_1\mathcal{N}_2}}-1\right).
\end{equation}
\end{proposition}
\begin{proof}
The proof is given in Appendix I. 
\end{proof}

For the special case of
$\mathcal{N}_2\sigma_f^2=\mathcal{N}_1\sigma_g^2$, it can be
seen from \eqref{5x} that $\displaystyle\lim_{\delta\rightarrow
0}\frac{1}{\delta}(\sqrt{1+\delta}-1)=\frac{1}{2}$, where $\delta=\frac{(\mathcal{N}_2\sigma_f^2-\mathcal{N}_1\sigma_g^2)\,P}{\mathcal{N}_1\sigma_g^2P+\mathcal{N}_1\mathcal{N}_2}
$. Hence, the optimum
$\tau$ is equal to $\frac{1}{2}$, which is in compliance with the
result obtained in \cite{jin06b} where assumed
$\mathcal{N}_1=\mathcal{N}_2$ and $\sigma_f^2=\sigma_g^2$.

\section{Performance Analysis}
\subsection{SER Expression of Relay Network with No CSI at the Source}
In the previous section, we have shown that the optimum transmitted
power of A$\&$F DSTC system based on maximizing the instantaneous
received SNR at the destination led to opportunistic relaying. In this section,
we will derive the SER formulas of best relay selection strategy
using A\&F. For this reason, we should first
derive the PDF of the received SNR at the destination due to the
$r$th relay, when other relays are silent, that is
\begin{align}\label{24k}
\gamma_r=\frac{P_1P_2\|\boldsymbol{f}_{r}\|^2\|\boldsymbol{g}_{r}\|^2}{P_2N_s
\|\boldsymbol{g}_{r}\|^2\!\mathcal{N}_1\!+\!N_s
N_d\mathcal{N}_2(\sigma_{f_r}^2P_1\!+\!\mathcal{N}_1\!)}.
\end{align}
Now, we will derive the PDF of $\gamma_r$, which is
required for calculating the average SER.
\begin{proposition}
For $\gamma_r$ in \eqref{24k}, the PDF
$p_{r}(\gamma)$ can be written as
\begin{align}\label{24}
    p_{r}(\gamma)&\!=\!\sum_{k=0}^{N_s}\!\frac{2a^k\overline{Y}_{\!r}^k\!{N_s \choose k}e^{-\frac{a
\gamma}{\overline{X}_{\!r}}}}{\gamma\,(N_d\!\!-\!\!1)!\,(N_s\!\!-\!\!1)!\,b_{r}^k}\!\left(\!\frac{b_{r}\gamma
}{\overline{X}_{\!r}\!\overline{Y}_{\!r}}\!\right)^{\!\!\!\frac{\mu+k}{2}}\!\!\!\!\!
    K_{\nu+k}\!\!\left(\!2\sqrt{\!\frac{b_{r}\gamma
}{\overline{X}_{\!r}\!\overline{Y}_{\!r}}}\right)\!,
\end{align}
where $\mu=N_s+N_d$, $\nu=N_d-N_s$,
$\overline{X}_{\!r}=N_s\sigma_{f_r}^2$,
$\overline{Y}_{\!r}=N_d\sigma_{g_r}^2$,
$a=\frac{N_s\mathcal{N}_1}{P_1}$, $b_r= \frac{N_s
N_d\mathcal{N}_2(\sigma_{f_r}^2P_1+\mathcal{N}_1)}{P_1P_2}$,
$K_n(x)$ is the modified Bessel function of the second kind of order
$n$.
\end{proposition}

\begin{proof}
The proof is given in Appendix II. 
\end{proof}

Define $\gamma_{\max}\triangleq\max\left\{\gamma_1,
\gamma_2,\ldots,\gamma_R\right\}$. The conditional SER of the best
relay selection system under A$\&$F mode
with \emph{R} relays can be written as 
\begin{equation}\label{25}
    P_e\left(R|\boldsymbol{F},\boldsymbol{G}\right)=c\,Q\left(\sqrt{g\,\gamma_{\max}}\right),
\end{equation}
where $Q(x)=1/\sqrt{2\pi}\int_x^\infty e^{-u^2/2}\,du$, and parameters $c$ and $g$ are represented as
\begin{equation*}
    c_{\text{QAM}}\!=\!4\frac{\sqrt{M}-1}{\sqrt{M}},\,\,c_{\text{PSK}}\!=\!2,\,\,g_{\text{QAM}}\!=\!\frac{3}{M-1},
    \,\,g_{\text{PSK}}\!=\!2\sin^2\!\!\left(\!\frac{\pi}{M}\!\right).
\end{equation*}

Using the result from order statistics, and by assuming that all
channel coefficients are independent of each other, the PDF of
$\gamma_{\max}$ can be written as
\begin{align}\label{28}
p_{\max}(\gamma)=\sum_{r=1}^{R}p_r(\gamma){\prod_{\underset{j\neq
r}{j=1}}^R}
    \text{Pr}\{\gamma_j<\gamma\},
\end{align}
where $\text{Pr}\{\gamma_{j}<\gamma\}$ can be evaluated as
\begin{align}\label{21c}
    &\text{Pr}\{\gamma_j<\gamma\}=\text{Pr}\{XY/(aY+b_j)<\gamma\}
    \nonumber\\
    &=\int_0^\infty
    \left(1-e^{-\frac{x}{\overline{X}_{\!j}}}\sum_{n=0}^{N_s-1}\frac{1}{n!}\left(\frac{x}{\overline{X}_{\!j}}\right)^{\!\!n}\right)\frac{y^{N_d-1}}{(N_d-1)!\,\overline{Y}_{\!j}^{N_d}}e^{-\frac{y}{\overline{Y}_{\!j}}}dy
    \nonumber\\
    &=1-\int_0^\infty
    \sum_{n=0}^{N_s-1}\left(\frac{\gamma(ay+b_j)}{y\overline{X}_j}\right)^{\!\!n}
    \frac{e^{-\frac{\gamma(ay+b_j)}{y\overline{X}_j}}y^{N_d-1}}{n!\,(N_d-1)!\,\overline{Y}_j^{N_d}}e^{-\frac{y}{\overline{Y}_j}}dy
    \nonumber\\
    &=1\!-\!\!\sum_{n=0}^{N_s\!-\!1}\!\sum_{k=0}^{n}\!\frac{{n \choose k}a^k b_{j}^{n-k}\gamma
    ^n
    e^{-\frac{a\gamma}{\overline{X}_{\!j}}}}{n!\,(N_d\!-\!1)!\,\overline{X}_{\!j}^n\overline{Y}_{\!j}^{N_d}}\!\int_0^\infty
    \!\!\!\!
    y^{N_d+k-n-1}e^{-\frac{y}{\overline{Y}_{\!j}}-\frac{b_{j}\gamma }{y\overline{X}_{\!j}}}dy
    \nonumber\\
    &=\!1\!-\!\!\!\sum_{n=0}^{N_s\!-\!1}\!\sum_{k=0}^{n}\!\!\frac{2{n \choose k}a^k \overline{Y}_{\!j}^{k}
    e^{-\frac{a\gamma}{\overline{X}_{\!j}}}}{n!\,(N_d\!-\!1)!\,b_{j}^{k}}\!\left(\!\frac{b_{j}\gamma
}{\overline{X}_{\!j}\!\overline{Y}_{\!j}}\!\right)^{\!\!\!\frac{N_d+n+k}{2}}\!\!\!\!\!\!\!\!\!
    K_{\!N_d-n+k}\!\!\left(\!2\sqrt{\!\frac{b_{j}\gamma
}{\overline{X}_{\!j}\!\overline{Y}_{\!j}}}\right)\!,
\end{align}
where $x=\frac{\gamma(ay+b_{j})}{y\overline{X}_{\!j}}$, and in the
second equality, we used the Erlang distribution \cite[Eq. (3.48)]{leo94}.

Now, we are deriving the SER expression for the selection relaying
scheme discussed in Section III. Averaging over conditional SER
$P_e\left(R|\boldsymbol{F},\boldsymbol{G}\right)$, we have the exact
SER expression as
\begin{align}\label{29}
    P_e(R)&=\int_{0}^\infty
    P_e\left(R|\boldsymbol{F},\boldsymbol{G}\right)\,p_{\max}(\gamma)\,d\gamma
    \nonumber\\
    &=\int_{0}^\infty
    c\,Q\left(\sqrt{g\,\gamma}\right)\,p_{\max}(\gamma)\,d\gamma.
\end{align}

\subsection{SER Expression of Relay Network with Partial CSI at the Source}
In this subsection,
we will derive the SER formulas of the best relay selection strategy
under the amplify-and-forward mode, when the source-relays CSI are available at the source. For this reason, we should first derive the PDF of the received SNR at the destination due to the $r$th relay, when other relays are silent. That is, from \eqref{e}, \eqref{24kkk}, and \eqref{e2}, we can rewrite
$\text{SNR}_{\text{ins}}=\sum_{r=1}^{R}\zeta_r$ as $\text{SNR}_{\text{ins}}=\max_{r\in \{1,\ldots,R\}}\zeta_r$ where
\begin{align}\label{24kk}
\zeta_r=\frac{P_1P_2\displaystyle\max_{n\in\{1,\ldots,N_s\}}|f_{n,r}|^2\|\boldsymbol{g}_{r}\|^2}{P_2
\|\boldsymbol{g}_{r}\|^2\!\mathcal{N}_1\!+\!
N_d\mathcal{N}_2(\sigma_{\xi_r}^2P_1\!+\!\mathcal{N}_1\!)}.
\end{align}

In the following, we will derive the PDF of $\zeta_r$, which is
required for calculating the average SER.
\begin{proposition}
For $\zeta_r$ in \eqref{24kkk}, the probability density function
$p_{\zeta_r}(\zeta)$ can be written as
\begin{align}\label{24x}
    p_{\zeta_r}(\zeta)\!=&\!\sum_{k=1}^{N_s}\!\frac{2(-1)^{k+1}\!{N_s \choose k}k \,\alpha\, e^{-\frac{\alpha
\zeta k}{\sigma_{f_r}^2}}}{(N_d\!\!-\!\!1)!\,\sigma_{f_r}^2}\!\left(\!\frac{\beta_{r}k\zeta
}{\sigma_{f_r}^2\overline{Y}_{\!r}}\!\right)^{\!\!\!\frac{N_d}{2}}\!\!\!\!\!
    K_{N_d}\!\!\left(\!2\sqrt{\!\frac{k \beta_{r}\zeta
}{\sigma_{f_r}^2\!\overline{Y}_{\!r}}}\right)\!
\nonumber\\
    &+\sum_{k=1}^{N_s}\!\frac{2(-1)^{k+1}\!{N_s \choose k} k \,\beta_r e^{-\frac{\alpha
\zeta k}{\sigma_{f_r}^2}}}{(N_d\!\!-\!\!1)!\,\sigma_{f_r}^2\overline{Y}_{\!r}}\!\left(\!\frac{\beta_{r}k\zeta
}{\sigma_{f_r}^2\overline{Y}_{\!r}}\!\right)^{\!\!\!\frac{N_d\!-1}{2}}\!\!\!\!\!
    K_{N_d\!-1}\!\!\left(\!2\sqrt{\!\frac{k \beta_{r}\zeta
}{\sigma_{f_r}^2\!\overline{Y}_{\!r}}}\right)\!,
\end{align}
where $\alpha=\frac{\mathcal{N}_1}{P_1}$ and $\beta_r= \frac{
N_d\mathcal{N}_2(\sigma_{\zeta_r}^2P_1+\mathcal{N}_1)}{P_1P_2}$.
\end{proposition}
\begin{proof}
The proof is given in Appendix III. 
\end{proof}

Let $\zeta_{\max}\triangleq\max\left\{\zeta_1,
\zeta_2,\ldots,\zeta_R\right\}$. The conditional SER of the best
relay selection system under A$\&$F mode
with \emph{R} relays and partial CSI at the source can be written as 
\begin{equation}\label{25b}
    P_e\left(R|\boldsymbol{F},\boldsymbol{G}\right)=c\,Q\left(\sqrt{g\,\zeta_{\max}}\right).
\end{equation}
From \eqref{28}, the PDF of $\zeta_{\max}$ can be written as
\begin{align}\label{28o}
p_{\zeta_{\max}}(\zeta)=\sum_{r=1}^{R}p_{\zeta_r}(\zeta){\prod_{\underset{j\neq
r}{j=1}}^R}
    \text{Pr}\{\zeta_j<\zeta\}.
\end{align}
where $\text{Pr}\{\zeta_{j}<\zeta\}$ can be evaluated by solving
the integral in the last equality of \eqref{21m} using \cite[Eq. (3.471)]{gra96} as
\begin{align}\label{24o}
    \text{Pr}\{\zeta_r<\zeta\}&=1+\sum_{k=1}^{N_s}\frac{2(-1)^k{N_s \choose k}\,e^{-\frac{\alpha
\zeta k}{\sigma_{f_r}^2}}}{(N_d\!\!-\!\!1)!}\!\left(\!\frac{\beta_{r}k\zeta
}{\sigma_{f_r}^2\overline{Y}_{\!r}}\!\right)^{\!\!\!\frac{N_d}{2}}\!\!\!\!\!
    K_{N_d}\!\!\left(\!2\sqrt{\!\frac{k \beta_{r}\zeta
}{\sigma_{f_r}^2\!\overline{Y}_{\!r}}}\right)\!.
\end{align}

Now, we are deriving the SER expression for the selection relaying
scheme discussed in Subsection III-C. Averaging over conditional SER
$P_e\left(R|\boldsymbol{F},\boldsymbol{G}\right)$, we have the exact
SER expression as
\begin{align}\label{29o}
    P_e(R)&=\int_{0}^\infty
    P_e\left(R|\boldsymbol{F},\boldsymbol{G}\right)\,p_{\zeta_{\max}}(\zeta)\,d\zeta
    \nonumber\\
    &=\int_{0}^\infty
    c\,Q\left(\sqrt{g\,\zeta}\right)\,p_{\zeta_{\max}}(\zeta)\,d\zeta.
\end{align}

\subsection{SER Expression of Relay Network with Antenna Selection at the Source and No CSI at the Relays}
Here, we study the performance analysis of the relaying scheme presented in Subsection III-C. From \eqref{16ss} and \eqref{e3}, we can write the instantaneous received SNR at the destination as $\text{SNR}_{\text{ins}}=\sum_{r=1}^{R}\eta_r$ where $\eta_r$ is defined as
\begin{equation}\label{35x}
    \eta_r=\frac{|f_{n^{\!*}\!,r}|^2\frac{P_2}{\sigma_{f_{n^{\!*}\!,r}}^2\!P_1\!+\!\mathcal{N}_1}
    \|\boldsymbol{g}_{r}\|^2}{N_d\displaystyle\sum_{k=1}^R\frac{P_2}{\sigma_{f_{n^{\!*}\!,k}}^2\!P_1\!+\!\mathcal{N}_1}\|
    \boldsymbol{g}_{k}\|^2\mathcal{N}_1+N_d\mathcal{N}_2}.
\end{equation}

Now, we can use the moment generating function (MGF) to derive the average SER expression for
the relay network discussed in Subsection III-C.
The conditional SER of the the A$\&$F DSTC with the antenna selection at the source can be given by
\begin{equation}\label{35w}
    P_e\left(R|\boldsymbol{F},\boldsymbol{G}\right)=c\,Q\!\left(\sqrt{g\,\sum_{r=1}^R\eta_r}\right).
\end{equation}
Since the $\eta_r$s are independent, the average SER would be
\begin{align}\label{36v}
P_e(R)&=\int_{0;R-\text{fold}}^\infty\!
    \!P_e\left(R|\boldsymbol{F},\boldsymbol{G}\right)\prod_{r=1}^R\!\left(p(\eta_r)\,d\eta_r\right)
=\int_{0;R-\text{fold}}^\infty\!
    \!c\,Q\!\!\left(\!\sqrt{g\,\sum_{r=1}^R\eta_r}\right)\prod_{r=1}^R\!\left(p(\eta_r)\,d\eta_r\right).
\end{align}
Using the moment generating function approach, we get
\begin{align}\label{37x}
P_e(R)&=\!\int_{0;\,R-\text{fold}}^\infty
\!\frac{c}{\pi}\!\!\int_0^{\frac{\pi}{2}} \!\!\prod_{r=1}^R
e^{-\frac{g\,\eta_r}{2\sin
^2\phi}}\,d\phi\prod_{r=1}^R\left(p(\eta_r)\,d\eta_r\right)
=\frac{c}{\pi}\!\!\int_0^{\frac{\pi}{2}}\prod_{r=1}^R
M_r\left(-\frac{g}{2\sin ^2\phi}\right)d\phi,
\end{align}
where $M_r(s)=\mathbb{E}_\gamma\{e^{s \eta_r}\}$ is the MGF of $\eta_r$ in \eqref{35x}.

\section{Asymptotic SER Expression}
Now, we are going to derive a closed-form SER formula at the
destination, which is valid in the high SNR regime.

\subsection{Asymptotic SER Expression of Relay Network with No CSI at the Source}
\subsubsection{Case of $N_d\neq N_s$} Here, a closed-form SER formula
for the case of $N_d\neq N_s$ is derived for high SNR scenarios.
This simple expression can be used for a power allocation strategy
among the cooperative nodes, or to get an insight on the
diversity-multiplexing tradeoff of the system.

Using the fact that $K_0(x)\approx-\ln\left(x\right)$ \cite[Eq.
(9.6.8)]{abr72}, as $x\rightarrow 0$, the $p_{r}(\gamma)$ in
\eqref{24} can be approximated as
\begin{align}\label{24d}
    p_{r}(\gamma)&\approx\sum_{k=0}^{N_s}\!\frac{{N_s \choose k}a^k \overline{Y}_{\!r}^{k}\Gamma(N_d-N_s+k)}{ (N_s\!-\!1)!\,(N_d\!-\!1)!\,b_{r}^{k}}
    \frac{e^{-\frac{a\gamma}{\overline{X}_{\!r}}}}{\gamma}\left(\!\frac{b_{r}\gamma
}{\overline{X}_{\!r}\overline{Y}_{\!r}}\!\right)^{\!N_s}
    \!,
\end{align}
for $N_d> N_s$, and using
$K_\nu(x)\approx\frac{1}{2}\Gamma(\nu)\left(\frac{x}{2}\right)^{-\nu}
, \,\nu\neq0$ \cite[Eq. (9.6.9)]{abr72}, as $x\rightarrow 0$, where
$\Gamma(\nu)$ is gamma function of order $\nu$, and
$K_\nu(x)=K_{-\nu}(x)$ \cite[Eq. (9.6.6)]{abr72}, for $N_d< N_s$,
$p_{r}(\gamma)$ can be approximated as
\begin{align}\label{24e}
    p_{r}(\gamma)\!\approx &\!\!\!\sum_{k=0}^{N_s-N_d-1}\!\!\frac{{N_s \choose k}a^k \overline{Y}_{\!r}^{k}\Gamma(N_s\!-\!N_d\!-\!k)}{ (N_s\!-\!1)!\,(N_d\!-\!1)!\,b_{r}^{k}}
    \frac{e^{-\frac{a\gamma}{\overline{X}_{\!r}}}}{\gamma}\left(\!\frac{b_{r}\gamma
}{\overline{X}_{\!r}\overline{Y}_{\!r}}\!\right)^{\!\!k+\!N_d}
    \nonumber\\
&-\frac{2{N_s \choose N_s\!-\!N_d}(a
\overline{Y}_{\!r})^{N_s\!-\!N_d}}{
(N_s\!-\!1)!\,(N_d\!-\!1)!\,b_{r}^{N_s\!-\!N_d}}\frac{e^{-\frac{a\gamma}{\overline{X}_{\!r}}}}{\gamma}\left(\!\frac{b_{r}\gamma
}{\overline{X}_{\!r}\overline{Y}_{\!r}}\!\right)^{\!\!N_s}\!\ln\!\left(\!2\sqrt{\frac{b_{r}\gamma
}{\overline{X}_{\!r}\overline{Y}_{\!r}}}\right)
    \nonumber\\
&\!+\!\!\!\!\sum_{k=N_s-N_d+1}^{N_s}\!\!\!\!\!\frac{{N_s \choose
k}a^k \overline{Y}_{\!r}^{k}\Gamma(N_d\!-\!N_s\!+\!k)}{
(N_s\!-\!1)!\,(N_d\!-\!1)!\,b_{r}^{k}}\frac{e^{-\frac{a\gamma}{\overline{X}_{\!r}}}}{\gamma}\left(\!\frac{b_{r}\gamma
}{\overline{X}_{\!r}\overline{Y}_{\!r}}\!\right)^{\!\!N_s}
    \!\!.
\end{align}

Before deriving the asymptotic expression for SER, we present two
lemmas.

\begin{lemma}
Let $N\triangleq\min\left\{N_{s},N_{d}\right\}$ and $N_{s}\neq
N_{d}$. The ($N-1$)th order derivative of $p_{r}(\gamma)$ with
respect to $\gamma$ at zero is computed as
\begin{align}\label{8e}
    \Phi_{N_s,N_d,r}\triangleq\frac{\partial^{N-1}
p_r}{\partial\gamma^{N-1}}(0)&=\left\{\begin{array}{cc}
                                                \displaystyle \sum_{k=0}^{N_s}\frac{{N_s \choose k}a^k b_{r}^{N_s-k}\,\Gamma(N_d\!-\!N_s\!+k)}
                                                { (N_d\!-\!1)!\,\overline{X}_{\!r}^{\!N_s}\overline{Y}_{\!r}^{N_s-k}}, & \text{if}\,\, N\!=\!N_s,
\\
\frac{\Gamma(N_s-N_d)}{ (N_s\!-\!1)!}\left(\!\frac{b_{r}
}{\overline{X}_{\!r}\overline{Y}_{\!r}}\!\right)^{\!N_d}, & \text{if}\,\, N\!=\!N_d.
                                              \end{array}
    \right.
\end{align}
Furthermore, the $n$th ($n < N-1$) order derivatives of
$p_{r}(\gamma)$ with respect to $\gamma$ at zero are null.
\end{lemma}
\begin{proof}
By applying the chain rule for differentiating composite functions
into $p_{r}(\gamma)$ in \eqref{24d}-\eqref{24e}, the desired result
in \eqref{8e} is obtained. The second part of the lemma can
straightforwardly be calculated using the same procedure.
\end{proof}

\begin{lemma}
All the derivatives of the PDF of $\gamma_{\max}$, i.e., $p_{\max}$,
evaluated at zero up to order $(NR-1)$ are zero, while the $NR$-th
order derivative is given by
\begin{equation}\label{10g}
\frac{\partial^{\,NR}
p_{\max}}{\partial\gamma^{NR}}(0)=R\prod_{r=1}^{R}\frac{\partial^{N-1}
p_r}{\partial\gamma^{N-1}}(0) ,
\end{equation}

\end{lemma}
\begin{proof}
Since $\gamma_r$ has non-negative values, it is obvious that
$\text{Pr}\{\gamma_r<0\}=0$. Therefore, using \eqref{28} and Lemma
1, and by applying the chain rule differentiating composite
functions, it can be shown that the derivatives of the PDF of
$p_{\max}$, evaluated at zero up to order $(NR-1)$ are zero. In
addition, $\frac{\partial^{\,NR} p_{\max}}{\partial\gamma^{NR}}(0)$
has a limited nonzero value when $N_{s}\neq N_{d}$ given by
\eqref{8e}, which completes the proof.
\end{proof}

An asymptotic expression for the SER of the system is presented in the
following proposition.
\begin{proposition}
Suppose the relay network consisting of $R$ relays and multiple
antenna source and destination. The SER of this system at high SNRs
can be approximated as
\begin{equation}\label{17e}
    P_e(R)\approx\frac{\displaystyle\prod_{i=1}^{NR+1}(2i-1)}{2(NR+1)g^{NR+1}}\frac{cR}{(NR)!}
\prod_{r=1}^{R}\Phi_{N_s,N_d,r}.
\end{equation}
\end{proposition}
\hspace{-.2cm}\begin{proof} To deduce the asymptotic behavior of the
average SER, we are using the approximate expression given in
\cite{wan03}. When the derivatives of $p_{\max}(\gamma)$ up to
$(NR-1)$-th order are null at $\gamma\!=\!0$,
then the SER at high SNRs can be given by
\begin{equation}\label{16f}
    P_e(R)\approx\frac{\displaystyle\prod_{i=1}^{{NR}+1}(2i-1)}{2({NR}+1)g^{{NR}+1}}\frac{c}{(NR)!}\frac{\partial^{NR}
    p_{\max}}{\partial\gamma^{NR}}(0).
\end{equation}
Applying Lemmas 2, we have
\begin{equation}\label{17f}
    P_e(R)\approx\frac{\displaystyle\prod_{i=1}^{NR+1}(2i-1)}{2(NR+1)g^{NR+1}}\frac{cR}{(NR)!}
\prod_{r=1}^{R}\frac{\partial^{N-1} p_r}{\partial\gamma^{N-1}}(0).
\end{equation}
Combining \eqref{8e} and \eqref{17f}, leads to \eqref{17e}.
\end{proof}

\begin{corollary}
The A$\&$F opportunistic relaying scheme with multiple antennas source
and destination over Rayleigh fading provides the diversity gain of
$R\min\{N_s,N_d\}$.
\end{corollary}

\begin{proof}
A tractable definition of the diversity gain is \cite[Eq.
(1.19]{jaf05}
\begin{equation}\label{40}
    G_d=-\lim_{\mu\rightarrow\infty}\frac{\log \left(P_e(R)\right)}{\log
    \left(\mu\right)},
\end{equation}
where $\mu$ denotes the transmit SNR. Now, using \eqref{8e} and
\eqref{17e}, it can be shown that $G_d=-\lim_{\mu\rightarrow\infty}\frac{\log \left(\prod_{r=1}^{R}b_r^{N}\right)}{\log
    \left(\mu\right)}=NR$, and thus, the diversity order $G_d$
becomes $R\min\{N_s,N_d\}$.
%
\end{proof}

\subsubsection{Case of $N_d= N_s$}
Here, we derive a tight upper-bound on the average SER of the system studied in Subsection III-A.
Since the $\gamma_r$s are independent, using \eqref{28} and
\eqref{29}, the average SER would be
\begin{align}\label{23s}
P_e(R)&\!=\!\sum_{r=1}^{R}\!\int_{0}^\infty\!\!\!
    c\,Q\left(\sqrt{g\,\gamma}\right)p_r(\gamma){\prod_{\underset{j\neq r}{j=1}}^R}\!
    \text{Pr}\{\gamma_j\!\leq\!\gamma\}\,d\gamma
    \leq\!\sum_{r=1}^{R}\!\int_{0}^\infty\!\!
    e^{-g\,\gamma}p_r(\gamma){\prod_{\underset{j\neq r}{j=1}}^R}\!
    \text{Pr}\{\gamma_j\!\leq\!\gamma\}\,d\gamma,
\end{align}
where in the inequality, we have used Chernoff bound $Q(x)\leq
e^{-x^2/2}$.

Using the facts that $K_0(x)\approx-\ln\left(x\right)$ \cite[Eq.
(9.6.8)]{abr72},
$K_\nu(x)\approx\frac{1}{2}\Gamma(\nu)\left(\frac{x}{2}\right)^{-\nu}
, \,\nu\neq0$ \cite[Eq. (9.6.9)]{abr72}, as $x\rightarrow 0$, and
$K_\nu(x)=K_{-\nu}(x)$ \cite[Eq. (9.6.6)]{abr72}, for the case of
$N_d= N_s$, the $p_{r}(\gamma)$ in \eqref{24} can be approximated as
\begin{align}\label{24s}
    p_{r}\!(\gamma)\!\approx\!&\sum_{k=1}^{N_s}\!\frac{{N_s \choose k}a^k \overline{Y}_{\!r}^{k}\Gamma(k)}{ (N_s\!-\!1)!\,(N_s\!-\!1)!\,b_{r}^{k}
    }\frac{e^{-\frac{a\gamma}{\overline{X}_{\!r}}}}{\gamma}\left(\!\frac{b_{r}\gamma
}{\overline{X}_{\!r}\overline{Y}_{\!r}}\!\right)^{\!N_s}
\!\!\!
\nonumber\\
&-\!\frac{2}{
(N_s\!-\!1)!\,(N_s\!-\!1)!}\frac{e^{-\frac{a\gamma}{\overline{X}_{\!r}}}}{\gamma}\!\left(\!\!\frac{b_{r}\gamma
}{\overline{X}_{\!r}\overline{Y}_{\!r}}\!\right)^{\!\!N_s}\!\!\!\ln\!\left(\!2\sqrt{\frac{b_{r}\gamma
}{\overline{X}_{\!r}\overline{Y}_{\!r}}}\right)
    \!.
\end{align}

To get a closed-form solution for the SER, using an upper-bound on
$\text{Pr}\{\gamma_r<\gamma\}$ in \eqref{21c}, i.e., $\text{Pr}\{\gamma_r<\gamma\}\leq 1-e^{-\frac{a\gamma}{\overline{X}_r}}$, we have
\begin{align}\label{40a}
P_e(R)&\!\leq\!\sum_{r=1}^{R}\!\int_{0}^\infty\!\!\!
    e^{-g\,\gamma}p_r(\gamma){\prod_{\underset{j\neq r}{j=1}}^R}\!
    \left(1-e^{-\frac{a\gamma}{\overline{X}_j}}\right)\!d\gamma
    \!\leq\!\sum_{r=1}^{R}\!\int_{0}^\infty\!\!
    e^{-g\,\gamma}p_r(\gamma)\gamma^{R-1}d\gamma \prod_{\underset{j\neq r}{j=1}}^R\!
    \left(\frac{a}{\overline{X}_{\!j}}\right).
\end{align}
Then, by replacing the $p_{r}(\gamma)$ in \eqref{24s} into
\eqref{40a}, an upper-bound on $P_e(R)$ is given by
\begin{align}\label{40c}
P_e(R) &\leq \sum_{r=1}^{R}\prod_{\underset{j\neq r}{j=1}}^R\!
    \left(\frac{a}{\overline{X}_{\!j}}\right)\left\{\sum_{k=1}^{N_s}\Psi_{r,k}\int_{0}^\infty
    e^{-g\,\gamma-\frac{a\gamma}{\overline{X}_{\!r}}} \gamma^{N_s+R-2}
    \,d\gamma\right.
\nonumber\\
    &\left.- \Psi_{r,0}\int_{0}^\infty
    2e^{-g\,\gamma-\frac{a\gamma}{\overline{X}_{\!r}}} \gamma^{N_s+R-2} \ln\!\left(2\sqrt{\frac{b_{r}\gamma
}{\overline{X}_{\!r}\overline{Y}_{\!r}}}\right)\,d\gamma\right\},
\end{align}
where $\Psi_{r,0}=\frac{ b_r^{N_s}}{
(N_s\!-\!1)!\,(N_s\!-\!1)!\,\overline{X}_r^{N_s}\overline{Y}^{N_s}}$
and
$\Psi_{r,k}=\frac{{N_s \choose k}a^k b_r^{N_s-k}\,\Gamma(k)}{
(N_s\!-\!1)!\,(N_s\!-\!1)!\,\overline{X}_r^{N_s}\overline{Y}^{N_s-k}},
\,\,\text{for}\,\,k=1,\ldots,N_s$.

Using \cite[Eq. (3.351)]{gra96} and \cite[Eq. (4.352)]{gra96},
\eqref{40c} can be calculated as
\begin{align}\label{41a}
&P_e(R)\leq\sum_{r=1}^{R}\left(g+\frac{a}{\overline{X}_{\!r}}\right)^{\!\!\!-N_s-R+1}\!(N_s+R-2)!\prod_{\underset{j\neq
r}{j=1}}^R\!
    \left(\frac{a}{\overline{X}_{\!j}}\right)
\nonumber\\
&\left\{\sum_{k=1}^{N_s}\Psi_{r,k}+\Psi_{r,0}\left[\ln\!\left(\frac{g\overline{X}_{\!r}\overline{Y}_{\!r}
+a\overline{Y}_{\!r}}{4b_{r}}\right)\!+\!\kappa-\!\!\!\!\sum_{i=1}^{N_s+R-2}\!\!\frac{1}{i}\right]\right\}.
\end{align}

Furthermore, for two cases of $R=1$ and $R=2$, we can find tighter
upper-bounds for SER as follows. First, when $R=1$, the second
inequality in \eqref{40a} becomes equality. For the case of $R=2$, by replacing $p_{r}(\gamma)$ from \eqref{24s}
into the first inequality in \eqref{40a}, we have
\begin{align}\label{44c}
&P_e(2) \leq
\sum_{r=1}^{2}\left\{\sum_{k=1}^{N_s}\Psi_{r,k}\int_{0}^\infty
    \left(e^{-g\,\gamma-\frac{a\gamma}{\overline{X}_{\!r}}}-e^{-g\,\gamma-\frac{2a\gamma}{\overline{X}_{\!r}}}\right) \gamma^{N_s-1}
    \,d\gamma\right.
\nonumber\\
    &\left.-\Psi_{r,0}\int_{0}^\infty
    2\left(e^{-g\,\gamma-\frac{a\gamma}{\overline{X}_{\!r}}}-e^{-g\,\gamma-\frac{2a\gamma}{\overline{X}_{\!r}}}\right) \gamma^{N_s-1} \ln\!\left(2\sqrt{\frac{b_{r}\gamma
}{\overline{X}_{\!r}\overline{Y}_{\!r}}}\right)\,d\gamma\right\}.
\end{align}
Then, similar to \eqref{41a}, a closed-form upper-bound for $P_e(2)$
can be calculated as
\begin{align}\label{44a}
&P_e(2)\leq\sum_{r=1}^{2}\left(g+\frac{a}{\overline{X}_{\!r}}\right)^{\!\!\!-N_s}\!(N_s-1)!
\left\{\sum_{k=1}^{N_s}\Psi_{r,k}+\Psi_{r,0}\left[\ln\!\left(\frac{g\overline{X}_{\!r}\overline{Y}_{\!r}
+a\overline{Y}_{\!r}}{4b_{r}}\right)\!+\!\kappa-\!\!\!\!\sum_{i=1}^{N_s-1}\!\!\frac{1}{i}\right]\right\}
\nonumber\\
&\,\,\,\,\,\,\,\,\,\,\,\,\,\,\,-\sum_{r=1}^{2}\left(g+\frac{2a}{\overline{X}_{\!r}}\right)^{\!\!\!-N_s}\!(N_s-1)!
\left\{\sum_{k=1}^{N_s}\Psi_{r,k}+\Psi_{r,0}\left[\ln\!\left(\frac{g\overline{X}_{\!r}\overline{Y}_{\!r}
+2a\overline{Y}_{\!r}}{4b_{r}}\right)\!+\!\kappa-\!\!\!\!\sum_{i=1}^{N_s-1}\!\!\frac{1}{i}\right]\right\}.
\end{align}

\subsection{Asymptotic SER Expression of Relay Network with Partial CSI at the Source}
Here, a closed-form SER formula of a relay network with partial CSI at the source, which is studied in Subsection III-B, is derived in high SNR scenarios, when $N_d> N_s$.
Before deriving the asymptotic expression for SER, we present two
lemmas.

\begin{lemma}
The ($N_s-1$)th order derivative of $p_{\zeta_r}(\zeta)$ with
respect to $\zeta$ at zero, when $N_d> N_s$, is computed as
\begin{align}\label{24pp}
    \frac{\partial^{N_s-1}
p_{\zeta_r}}{\partial\zeta^{N_s-1}}(0)\!&=\sum_{k=1}^{N_s}\frac{\alpha^{N_s-k}\beta_r^k N_s! {N_s \choose k}(N_d\!-k\!-\!1)!}{(N_d\!-\!1)! \overline{Y}_r^{\,-k}}\triangleq \Delta_{N_s,N_d,r}.
\end{align}
Furthermore, the $n$th ($n < N_s-1$) order derivatives of
$p_{\zeta_r}(\zeta)$ with respect to $\zeta$ at zero are null.
\end{lemma}
\begin{proof}
The proof is given in Appendix IV.
\end{proof}

\begin{lemma}
All the derivatives of the PDF of $\zeta_{\max}$, i.e., $p_{\zeta_{\max}}$,
evaluated at zero up to order $(N_sR-1)$ are zero, while the $N_sR$-th
order derivative is given by
\begin{equation}\label{10k}
\frac{\partial^{\,N_sR}
p_{\zeta_{\max}}}{\partial\zeta^{N_sR}}(0)=R\prod_{r=1}^{R}\frac{\partial^{N_s-1}
p_{\zeta_r}}{\partial\zeta^{N_s-1}}(0).
\end{equation}

\end{lemma}
\begin{proof}
Since $\zeta_r$ is non-negative,
$\text{Pr}\{\zeta_r<0\}=0$. Therefore, using \eqref{28o} and Lemma 3, and by applying the chain rule differentiating composite functions, it can be shown that the derivatives of the PDF of
$p_{\zeta_{\max}}$, evaluated at zero up to order $(N_sR-1)$ are zero. In
addition, $\frac{\partial^{\,N_sR} p_{\zeta_{\max}}}{\partial\zeta^{N_sR}}(0)$
has a limited non-zero value when $N_{s}< N_{d}$ given by
\eqref{24pp}, which completes the proof.
\end{proof}

Asymptotic expression for the SER of the system is presented in the
following proposition:
\begin{proposition}
Suppose a relay network consisting of $R$ relays with multiple
antenna source and destination. The SER of this system at high SNRs
can be calculated as
\begin{equation}\label{17ee}
    P_e(R)\approx\frac{\displaystyle\prod_{i=1}^{N_sR+1}(2i-1)}{2(N_sR+1)g^{N_sR+1}}\frac{cR}{(N_sR)!}
\prod_{r=1}^{R}\Delta_{N_s,N_d,r}.
\end{equation}
\end{proposition}
\hspace{-.2cm}\begin{proof} To deduce the asymptotic behavior of the
average SER, we are using the approximate expression given in
\cite{wan03}. When the derivatives of $p_{\zeta_{\max}}(\zeta)$ up to
$(N_sR-1)$-th order are null at $\zeta\!=\!0$,
then the SER at high SNRs can be given by \eqref{16f}.
Applying Lemmas 2, we have
\begin{equation}\label{17ff}
    P_e(R)\approx\frac{\displaystyle\prod_{i=1}^{N_sR+1}(2i-1)}{2(N_sR+1)g^{N_sR+1}}\frac{cR}{(N_sR)!}
\prod_{r=1}^{R}\frac{\partial^{N_s-1} p_{\zeta_r}}{\partial\zeta^{N_s-1}}(0).
\end{equation}
Combining \eqref{24pp} and \eqref{17ff}, \eqref{17ee} is obtained.
\end{proof}

\begin{corollary}
The A$\&$F opportunistic relaying scheme with multiple antennas source
and destination over Rayleigh fading channels provides the diversity gain of
$N_sR$, when $N_s<N_d$.
\end{corollary}

\begin{proof}
Using \eqref{40}, \eqref{24pp}, and
\eqref{17ee}, it is easy to show that $G_d=-\displaystyle\lim_{\mu\rightarrow\infty}\frac{\log \left(\alpha^{N_sR}\right)}{\log
    \left(\mu\right)}=N_sR$, and thus, the diversity order $G_d$
becomes $N_sR$.
%
\end{proof}

\begin{figure}[e]
  \centering
  \includegraphics[width=\columnwidth]{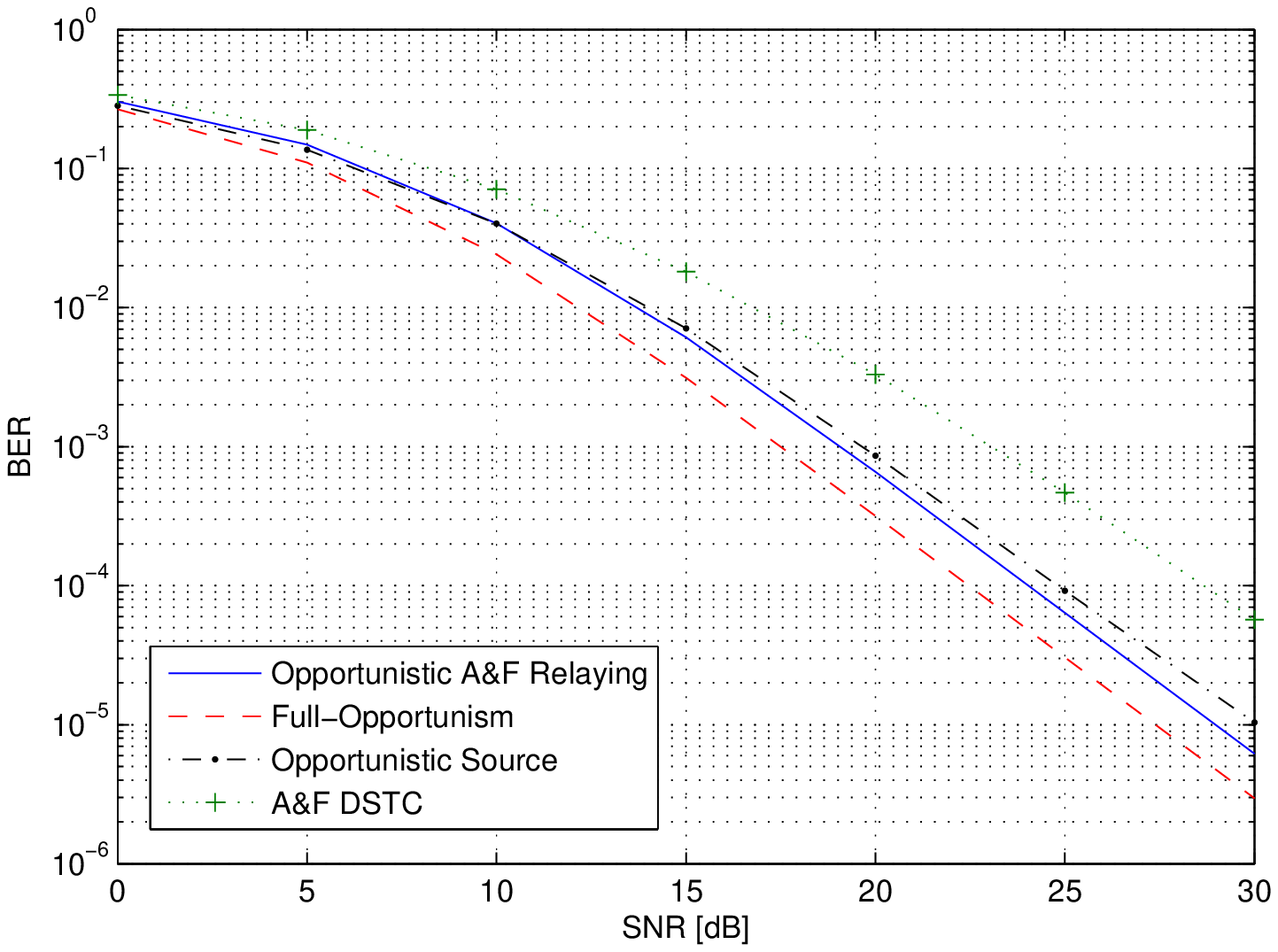}\\
  \caption{Performance comparison of A\&F DSTC with multiple antenna source with the proposed opportunistic relaying schemes for a relay network with BPSK signals, $R=2$, $N_s=2$, $N_d=1$, and $T=4$.
  }
  \label{f2}
\end{figure}
\begin{figure}[e]
  \centering
  \includegraphics[width=\columnwidth]{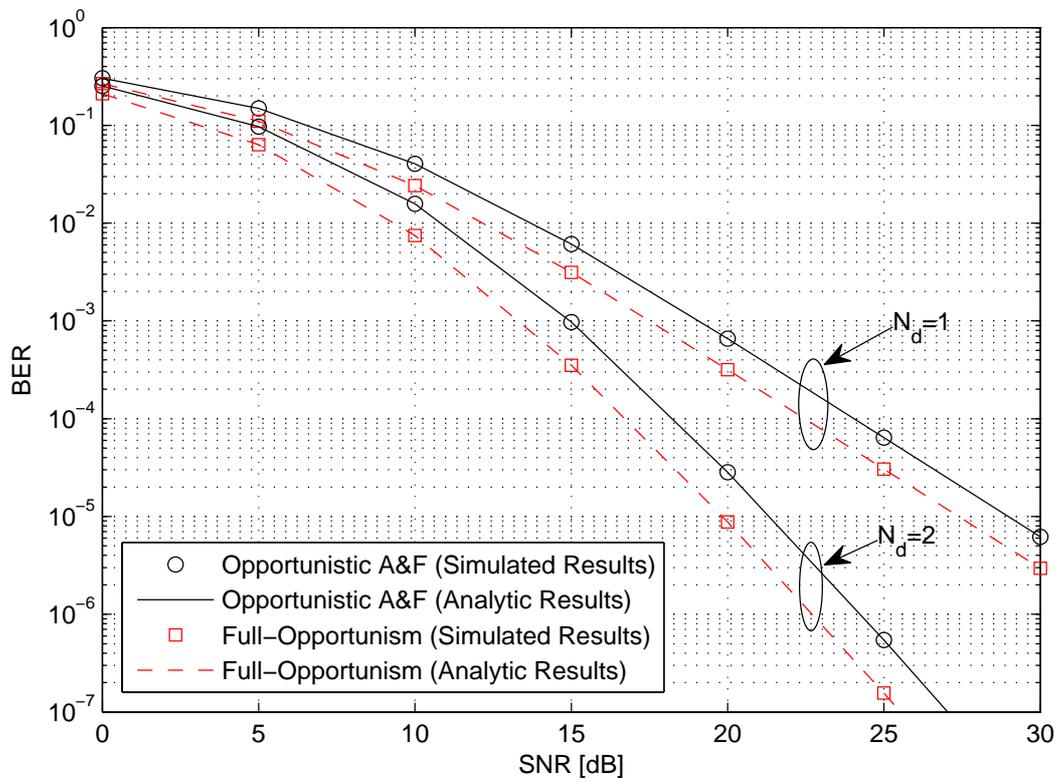}\\
  \caption{Performance comparison of analytical and simulated results
  of a relay network with BPSK signals, $R=2$, $N_s=2$, and $T=4$.
  }
  \label{f3}
\end{figure}
\begin{figure}[e]
  \centering
  \includegraphics[width=\columnwidth]{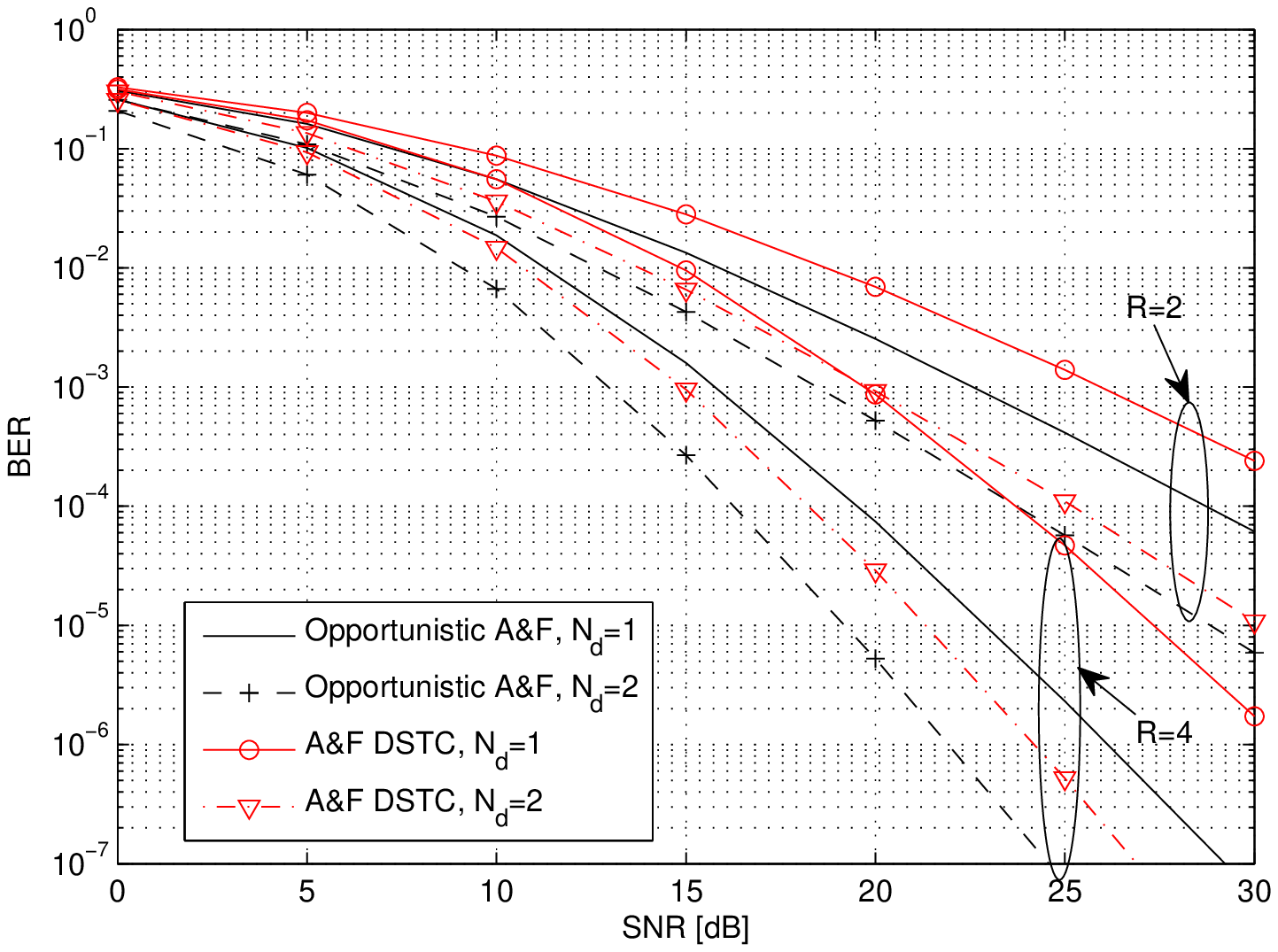}\\
  \caption{Performance comparison of A\&F DSTC and opportunistic A\&F relaying for a relay network for different values of relays and the destination antennas with BPSK signals, $N_s=1$, and $T=4$.
  }
  \label{f4}
\end{figure}

\section{Simulation Results}
In this section, the performance of distributed space-time codes are
compared with opportunistic relaying schemes in A$\&$F mode presented in Section III.
The signal symbols are modulated as BPSK. We fixed the total power
consumed in the whole network as $P$ and use the equal power
allocation, i.e, $P_1=P_2=\frac{P}{2}$. Assume that the relays and the
destination have the same value of noise power, i.e.,
$\mathcal{N}_1=\mathcal{N}_2$, and all the links have unit-variance
Rayleigh flat fading, i.e., $\sigma_{f_r}^2=\sigma_{g_r}^2=1$. Let $T=4$, and we use the orthogonal
space-time code structure in \eqref{1} and \eqref{6a}. For the case
of $N_s=2$, $R=2$, the matrices $\boldsymbol{A}_1$ and $\boldsymbol{A}_2$ used at the source and the matrices $\boldsymbol{C}_1$ and $\boldsymbol{C}_2$ used at the relays are as follows:
$\boldsymbol{A}_1=\boldsymbol{C}_1=\boldsymbol{I}_4$, and
\begin{align}\label{20g}
&\boldsymbol{A}_2=\left[\begin{array}{cccc}
                                     0 & -1 & 0 & 0 \\
                                     1 & 0 & 0 & 0 \\
                                     0 & 0 & 0 & -1 \\
                                     0 & 0 & 1 & 0
                                   \end{array}
                        \right],
     { \, }\boldsymbol{C}_2=\left[\begin{array}{cccc}
                                     0 & 0 & -1 & 0 \\
                                     0 & 0 & 0 & 1 \\
                                     1 & 0 & 0 & 0 \\
                                     0 & -1 & 0 & 0
                                   \end{array}
                        \right]\!\!.
\end{align}

In Fig.~\ref{f2}, the BER performance of the A$\&$F DSTC is compared to
the proposed opportunistic A$\&$F relaying schemes, when the number of available
relays is 2. For A$\&$F DSTC, equal power allocation is used among the
relays \cite{jin05}. The opportunistic A$\&$F scheme is based on the power allocation presented in Subsection III-A, in which the distributed space-time code is applied across the source antennas in the first phase and the best relay is selected in the second phase of transmission. In another scheme, called full-opportunism, we use the power allocation derived in Subsection III-B, in which the CSI is available for the maximum SNR power allocation across the source's antennas and the relays. The other scheme which is called opportunistic source and studied in Subsection III-C, uses the best antenna selection at the source and the distributed space-time code across the relays. One can observe from Fig. \ref{f2} that  full-opportunism outperforms the opportunistic A$\&$F relaying
scheme by more than 1.5 dB SNR at BER $10^{-4}$. Moreover, the opportunistic A$\&$F relaying and opportunistic source schemes achieve around 5 dB and 4 dB gain in SNR over A$\&$F DSTC at BER $10^{-4}$. Observing the curves behavior in high SNR, it can be seen that the diversity order of the system agrees with $R\min\{N_s, N_d\}$.

Fig.~\ref{f3} confirms that the analytical results attained in Section IV for finding the average SER for opportunistic A$\&$F relaying with space-time coded source and also full-opportunism scheme have the same performance as the
simulation results. In Fig.~\ref{f3}, we consider a network with $R=2$ and $N_s=2$ and two values of $N_d\in\{1,2\}$.
The analytical results are based on \eqref{29} and \eqref{29o} for opportunistic A$\&$F relaying and full-opportunism, respectively. It is shown that full-opportunism outperforms opportunistic A$\&$F relaying around 1.5 dB gain in SNR at BER $10^{-5}$ for both cases of $N_d=1$ and $N_d=2$.

In Fig.~\ref{f4}, the performance of A$\&$F DSTC and opportunistic A$\&$F relaying systems are compared for two values of relay numbers $R=2,4$, when a single antenna is used at the source, i.e., $N_s=1$, and $N_d=1,2$. Since it is assumed $N_s=1$, opportunistic A$\&$F relaying and full-opportunism have the same performance. In addition, A$\&$F DSTC and opportunistic source schemes become equivalent. Observing the curves behavior at high SNR, it can be seen that the diversity order of the system becomes $R\min\{N_s, N_d\}$. Furthermore, it can be seen that the performance of A$\&$F DSTC in low SNR conditions degrades as the number of relays increases due to the noise accumulation in the relays. For example, although A$\&$F DSTC system with $R=4$, $N_s=1$, and $N_d=1$ achieve the diversity gain of 4 in comparison to A$\&$F DSTC system with $R=2$, $N_s=1$, and $N_d=2$ with the diversity gain of 2, the former outperforms the latter about 2.5 dB at BER $10^{-2}$.

\section{Conclusion} In this paper, we studied the problem of power allocation and coding for a wireless relay network where both the transmitter and receiver have several antennas, while each relay has one. Due to the high transmission rate, it is not assumed relays are able to decode, and thus, a distributed space-time scheme is used, where relays just do a simple operation on the received signal before forwarding it. The optimal transmit power from the source antennas and the relays in the sense of maximizing the SNR at the destination are derived for a A$\&$F wireless relay network with multiple antenna terminals. Based on the knowledge of CSI at the source and the relay, we have derive three transmission schemes. We analyzed the average SER performance of the A\&F opportunistic relaying and full-opportunism
systems with $M$-PSK and $M$-QAM signals. Simulations are in
accordance with the analytic expressions. We also studied the asymptotic behavior of the proposed schemes and derived the closed-form SER formulas in the high SNR regime.

\appendices

\section{Proof of Proposition 1}
The average SNR at the destination can be obtained by dividing the
average received signal power by the variance of the noise at the
destination (approximation of
$\mathbb{E}\{\text{SNR}_{\text{ins}}\}$ using Jensen's inequality).
Using \eqref{15}, the average SNR can be written as
\begin{equation}\label{6x}
\text{SNR}=\frac{\tau(1-\tau)P^2\sigma_f^2\sigma_g^2}{\tau(\mathcal{N}_2\sigma_f^2-\mathcal{N}_1\sigma_g^2)P
+\mathcal{N}_1\sigma_g^2P+\mathcal{N}_1\mathcal{N}_2},
\end{equation}
where we have assumed $\sigma_{f_r}^2=\sigma_{f}^2$ and
$\sigma_{g_r}^2=\sigma_{g}^2$, for $r=1, \ldots, R$, and thus,
$P_{2,r}=\frac{P_2}{R}$. First, we consider the case in which
$\mathcal{N}_2\sigma_f^2\geq\mathcal{N}_1\sigma_g^2$. In this case, the
optimum value of $\tau$ which maximizes \eqref{6x}, subject to the
constraint $0<\tau<1$, is obtained as
\begin{equation}\label{7x}
\tau=\frac{\sqrt{1+\delta}-1}{\delta},
\end{equation}
where
\begin{equation}\label{8x}
\delta=\frac{(\mathcal{N}_2\sigma_f^2-\mathcal{N}_1\sigma_g^2)\,P}{\mathcal{N}_1\sigma_g^2P+\mathcal{N}_1\mathcal{N}_2}.
\end{equation}

Similarly, when $\mathcal{N}_2\sigma_f^2<\mathcal{N}_1\sigma_g^2$,
the optimum value of $\tau$, which maximizes SNR in \eqref{6x},
subject to constraint $0<\tau<1$, is also \eqref{7x} and
\eqref{8x}. Therefore, observing \eqref{7x} and \eqref{8x}, the
desired result in \eqref{5x} is achieved.

\section{Proof of Proposition 2}
Suppose $X=\|\boldsymbol{f}_{r}\|^2$ and $Y=\|\boldsymbol{g}_{r}\|^2$, where $X$ and
$Y$ have gamma distribution \cite[Eq. (5.14)]{sim00} with mean of
$\overline{X}_{\!r}$ and $\overline{Y}_{\!r}$, respectively.
Therefore, the cumulative density function (CDF) of
$\gamma_r=XY/(aY+b_r)$ can be presented to be
\begin{align}\label{21v}
    \text{Pr}\{\gamma_r<\gamma\}&=\text{Pr}\{XY/(aY+b_r)<\gamma\}
=\int_0^\infty
    \text{Pr}\left\{X<\frac{\gamma(ay+b_r)}{y}\right\}p_Y(y)dy
    \nonumber\\
    &=\int_0^\infty\!\!
    \left(1\!-\!\frac{\Gamma\left(N_s,\frac{\gamma(ay+b_r)}{y\overline{X}_r}\right)}
    {\Gamma(N_s)}\right)\frac{y^{N_d-1}e^{-\frac{y}{\overline{Y}_r}}}{(N_d-1)!\,\overline{Y}_r^{N_d}}dy
    \nonumber\\
    &=1\!-\!\int_0^\infty\!
    \frac{\Gamma\left(N_s,\frac{\gamma(ay+b_r)}{y\overline{X}_r}\right)}
    {\Gamma(N_s)}\frac{y^{N_d-1}e^{-\frac{y}{\overline{Y}_r}}}{(N_d-1)!\,\overline{Y}_r^{N_d}}dy,
\end{align}
where we have used \cite[Eq. (3.324)]{gra96} for the third equality,
$\Gamma(\alpha,x)$ is the incomplete gamma function of order
$\alpha$ \cite[Eq. (8.350)]{abr72}, and
$p_Y(y)=\frac{y^{N_d-1}}{(N_d-1)!\overline{\gamma}_r^{N_d}}e^{-\frac{y}{\overline{Y}_r}}$
\cite[Eq. (5.14)]{sim00}.
Then, using \eqref{21v}, $\Gamma(N_s)=(N_s-1)!$, and
$\frac{-d\,{\Gamma}(\alpha,x)}{dx}=x^{\alpha-1}e^{-x}$ \cite[Eq.
(8.356)]{gra96}, the PDF of $\gamma_r$ can be written as
\begin{align}\label{26}
    p_{r}(\gamma)&=\frac{\overline{X}_{\!r}^{\,-N_s}\gamma^{N_s-1}e^{-\frac{a \gamma}{\overline{X}_{\!r}}}}{\overline{Y}_{\!r}^{N_d}(N_d-1)!\,(N_s-1)!}
    \int_0^\infty y^{N_d-N_s-1}(ay+b_{r})^{N_s}
    e^{-\left(\frac{b_{r}\gamma}{\overline{X}_{\!r}y}+\frac{y}{\overline{Y}_{\!r}}\right)}
    dy
    \nonumber\\
    &=\left(\frac{b_{r}}{\overline{X}_{\!r}}\right)^{\!\!N_s}\!\!
    \frac{\gamma^{N_s-1}e^{-\frac{a \gamma}{\overline{X}_{\!r}}}}{\overline{Y}_{\!r}^{N_d}(N_d-1)!\,(N_s-1)!}
    \sum_{k=0}^{N_s}{N_s \choose k}
    \int_0^\infty y^{N_d-N_s-1}\left(\frac{ay}{b_{r}}\right)^{\!\!k}
    e^{-\left(\frac{b_{r}\gamma}{\overline{X}_{\!r}y}+\frac{y}{\overline{Y}_{\!r}}\right)}
    dy,
\end{align}
where we have used binomial theorem \cite[Eq. (2.36)]{leo94} in the
second equality. Thus, the PDF of $\gamma_r$ can be found by solving
the integral in \eqref{26} using \cite[Eq. (3.471)]{gra96}, yielding
\eqref{24}.

\section{Proof of Proposition 3}
Suppose $X=\displaystyle\max_{n\in\{1,\ldots,N_s\}}|f_{n,r}|^2$ and $Y=\|\boldsymbol{g}_{r}\|^2$ where
$Y$ has gamma distribution \cite[Eq. (5.14)]{sim00} with mean of
$\overline{Y}_{\!r}$. Therefore, the cumulative density function
(CDF) of $\zeta_r=XY/(\alpha Y+\beta_r)$ can be presented to be
\begin{align}\label{21m}
    \text{Pr}\{\zeta_r<\zeta\}&=\text{Pr}\{XY/(\alpha Y+\beta_r)<\zeta\}
    =\int_0^\infty
    \text{Pr}\left\{X<\frac{\zeta(\alpha y+\beta_r)}{y}\right\}p_Y(y)dy
    \nonumber\\
    &=\int_0^\infty
    \text{Pr}\{|f_{1,r}|^2<\frac{\zeta(\alpha y+\beta_r)}{y},\ldots,|f_{N_s,r}|^2<\frac{\zeta(\alpha y+\beta_r)}{y}\}p_Y(y)dy
    \nonumber\\
    &=\int_0^\infty
    \prod_{n=1}^{N_s}\text{Pr}\left\{|f_{n,r}|^2<\frac{\zeta(\alpha y+\beta_r)}{y}\right\}p_Y(y)dy
    \nonumber\\
    &
    =\int_0^\infty
    \left(1-e^{-\frac{\zeta(\alpha y+\beta_r)}{y\sigma_{f_r}^2}}\right)^{N_s}\frac{y^{N_d-1}}{(N_d-1)!\,\overline{Y}_r^{N_d}}
    e^{-\frac{y}{\overline{Y}_r}}dy
    \nonumber\\
    &=\sum_{k=0}^{N_s}{N_s \choose k}\int_0^\infty
    e^{-\frac{\zeta(\alpha y+\beta_r)k}{y\sigma_{f_r}^2}}\frac{(-1)^k y^{N_d-1}}{(N_d-1)!\,\overline{Y}_r^{N_d}}
    e^{-\frac{y}{\overline{Y}_r}}dy.
\end{align}
Then, the PDF of $\zeta_r$ can be written as
\begin{align}\label{21vm}
    p_{\zeta_r}(\zeta)=\frac{d}{d\zeta}\text{Pr}\{\zeta_r<\zeta\}
    =&\sum_{k=1}^{N_s}{N_s \choose k}\int_0^\infty
    e^{-\frac{\zeta(\alpha y+\beta_r)k}{y\sigma_{f_r}^2}}\frac{(-1)^{k+1}k\,\alpha  \, y^{N_d-1}}{(N_d-1)!\,\overline{Y}_r^{N_d}\sigma_{f_r}^2}
    e^{-\frac{y}{\overline{Y}_r}}dy
    \nonumber\\
    &+\sum_{k=1}^{N_s}{N_s \choose k}\int_0^\infty
    e^{-\frac{\zeta(\alpha y+\beta_r)k}{y\sigma_{f_r}^2}}\frac{(-1)^{k+1} k \, \beta_r\,  y^{N_d-2}}{(N_d-1)!\,\overline{Y}_r^{N_d}\sigma_{f_r}^2}
    e^{-\frac{y}{\overline{Y}_r}}dy,
\end{align}
and using \cite[Eq. (3.471)]{gra96}, the PDF of $\zeta_r$ yields to \eqref{24x}.

\section{Proof of Lemma 3}
For finding the value of $p_{\zeta_r}(\zeta)$ and its derivatives around zero, we use the fifth equation of \eqref{21m} to write
\begin{align}\label{21p}
    &p_{\zeta_r}(\zeta)\!=\!\frac{d}{d\zeta}\text{Pr}\{\zeta_r<\zeta\}
    \!=\!\int_0^\infty\!\!\!
    N_s\!\left(1\!-\!e^{-\frac{\zeta(\alpha y+\beta_r)}{y\sigma_{f_r}^2}}\right)^{\!N_s-1}\!e^{-\frac{\zeta(\alpha y+\beta_r)}{y\sigma_{f_r}^2}}
    \frac{(\alpha y+\beta_r) y^{N_d-2}}{\sigma_{f_r}^2(N_d-1)!\,\overline{Y}_r^{N_d}}
    e^{-\frac{y}{\overline{Y}_r}}dy.
\end{align}
Therefore, it follows from \eqref{21m} that $p_{\zeta_r}\!(0)=0$. Moreover, from \eqref{21m}, it can be seen that
$\frac{\partial^{n}
p_{\zeta_r}}{\partial\zeta^{n}}(0)=0$, for $n=1,\ldots, N_s-2$, and $\frac{\partial^{N_s-1}
p_{\zeta_r}}{\partial\zeta^{N_s-1}}(0)$ can be calculated as
\begin{align}\label{21k}
    &\frac{\partial^{N_s-1}
p_{\zeta_r}}{\partial\zeta^{N_s-1}}(0)=
    \lim_{\zeta\rightarrow 0}\int_0^\infty
    N_s!\left(\frac{\alpha y+\beta_r}{y\sigma_{f_r}^2}\right)^{\!\!N_s}e^{-\frac{N_s\zeta(\alpha y+\beta_r)}{y\sigma_{f_r}^2}}
    \frac{y^{N_d-1}}{(N_d-1)!\,\overline{Y}_r^{N_d}}
    e^{-\frac{y}{\overline{Y}_r}}dy
    \nonumber\\
    &=\lim_{\zeta\rightarrow 0}\sum_{k=1}^{N_s}\!\frac{\alpha^{N_s-k}\beta_r^k N_s! {N_s \choose k}e^{-\frac{N_s \alpha
\zeta}{\sigma_{f_r}^2}}}{(N_d\!\!-\!\!1)!\overline{Y}_r^{N_d}}\int_0^\infty
    e^{-\frac{N_s\zeta \beta_r}{y\sigma_{f_r}^2}-\frac{y}{\overline{Y}_r}}
    y^{N_d-k-1}dy
    \nonumber\\
    &=\lim_{\zeta\rightarrow 0}\sum_{k=1}^{N_s}\!\frac{2\alpha^{N_s-k}\beta_r^k N_s! {N_s \choose k}e^{-\frac{N_s \alpha
\zeta}{\sigma_{f_r}^2}}}{(N_d\!\!-\!\!1)!\overline{Y}_r^{N_d}}\left(\!\frac{\beta_{r}k\overline{Y}_{\!r}\zeta
}{\sigma_{f_r}^2}\!\right)^{\!\!\!\frac{N_d-k}{2}}\!\!\!\!\!
    K_{N_d-k}\!\!\left(\!2\sqrt{\!\frac{k \beta_{r}\zeta
}{\sigma_{f_r}^2\!\overline{Y}_{\!r}}}\right).
\end{align}
Using the fact that $K_\nu(x)\approx\frac{1}{2}\Gamma(\nu)\left(\frac{x}{2}\right)^{-\nu}
, \,\nu\neq0$ \cite[Eq. (9.6.9)]{abr72}, as $x\rightarrow 0$, $\frac{\partial^{N_s-1}
p_{\zeta_r}}{\partial\zeta^{N_s-1}}(0)$ in
\eqref{21k} can be approximated as
\begin{align}\label{24p}
    \frac{\partial^{N_s-1}
p_{\zeta_r}}{\partial\zeta^{N_s-1}}(0)\!&=\lim_{\zeta\rightarrow 0}\sum_{k=1}^{N_s}\!\frac{\alpha^{N_s-k}\beta_r^k N_s! {N_s \choose k}(N_d\!-k\!-\!\!1)!\,e^{-\frac{N_s \alpha
\zeta}{\sigma_{f_r}^2}}}{(N_d\!-\!1)! \overline{Y}_r^{\,-k}}
\nonumber\\
&=\sum_{k=1}^{N_s}\!\frac{\alpha^{N_s-k}\beta_r^k N_s! {N_s \choose k}(N_d\!-k\!-\!1)!}{(N_d\!-\!1)! \overline{Y}_r^{\,-k}},
\end{align}
which completes the proof.

\bibliographystyle{ieeetr}
\bibliography{references}

\begin{thebibliography}{10}

\bibitem{sen03a}
A.~Sendonaris, E.~Erkip, and B.~Aazhang, ``User cooperation diversity. {P}art
  {I}. {S}ystem description,'' {\em IEEE Trans. Commun.}, vol.~51, no. 11,
  pp.~1927--1938, Nov. 2003.

\bibitem{sen03b}
A.~Sendonaris, E.~Erkip, and B.~Aazhang, ``User cooperation diversity. {P}art
  {II}. {I}mplementation aspects and performance analysis,'' {\em IEEE Trans.
  Commun.}, vol.~51, pp.~1939--1948, Nov. 2003.

\bibitem{lan00}
J.~N. Laneman and G.~Wornell, ``Energy-efficient antenna sharing and relaying
  for wireless networks,'' in {\em Proc. Wireless Communications Networking
  Conf.}, (Chicago, IL), pp.~7--12, Sep. 2000.

\bibitem{lan02a}
J.~N. Laneman and G.~Wornell, ``Distributed space-time coded protocols for
  exploiting cooperative diversity in wireless networks,'' in {\em IEEE
  GLOBECOM 2002}, vol.~1, (Taipei, Taiwan, R.O.C.), pp.~77--81, Nov. 2002.

\bibitem{nab04}
R.~U. Nabar, H.~B{\"o}lcskei, and F.~W. Kneubuhler, ``Fading relay channels:
  {P}erformance limits and space-time signal design,'' {\em IEEE J. Sel. Areas
  Commun.}, vol.~22, no. 6, pp.~1099--1109, Aug. 2004.

\bibitem{hua03}
Y.~Hua, Y.~Mei, and Y.~Chang, ``Wireless antennas-making wireless
  communications perform like wireline communications,'' in {\em IEEE AP-S
  Topical Conf. on Wireless Comm. Tech.}, (Honolulu, Hawaii), Oct. 2003.

\bibitem{jin06b}
Y.~Jing and B.~Hassibi, ``Distributed space-time coding in wireless relay
  networks,'' {\em IEEE Trans. Wireless Commun.}, vol.~5, no. 12,
  pp.~3524--3536, Dec. 2006.

\bibitem{jin07}
Y.~Jing and H.~Jafarkhani, ``Using orthogonal and quasi-orthogonal designs in
  wireless relay networks,'' {\em IEEE Trans. Info. Theory}, vol.~53, no. 11,
  pp.~4106--4118, Nov. 2007.

\bibitem{mah09twc}
B.~Maham, A.~Hjørungnes, and G.~Abreu, ``Distributed {GABBA} space-time codes
  in amplify-and-forward relay networks,'' {\em IEEE Trans. Wireless Commun.},
  vol.~8, no. 4, pp.~2036--2045, Apr. 2009.

\bibitem{sus07}
G.~S. Rajan and B.~S. Rajan, ``Distributed space-time codes for cooperative
  networks with partial {CSI},'' in {\em Proc. IEEE Wireless Communications and
  Networking Conference (WCNC)}, (Hong Kong, China), pp.~902--906, March 2007.

\bibitem{jin05}
Y.~Jing and B.~Hassibi, ``Cooperative diversity in wireless relay networks with
  multiple-antenna nodes,'' in {\em IEEE Int. Symp. Inform. Theory}, (Adelaide,
  Australia), 2005.

\bibitem{pet08}
S.~Peters and R.~W. Heath, ``Selection cooperation in multi-source cooperative
  networks,'' {\em IEEE Signal Processing Letters}, vol.~15, pp.~421--424, Jan.
  2008.

\bibitem{ogg08}
F.~Oggier and B.~Hassibi, ``An algebraic coding scheme for wireless relay
  networks with multiple-antenna nodes,'' {\em IEEE Trans. Signal Process.},
  vol.~56, no. 7, pp.~2957--2966, Jul. 2008.

\bibitem{Hon07}
Y.-W. Hong, W.-J. Huang, F.-H. Chiu, and C.-C.~J. Kuo, ``Cooperative
  communications in resource-constrained wireless networks,'' {\em IEEE Signal
  Processing Magazine}, vol.~24, pp.~47--57, May 2007.

\bibitem{mah08v}
B.~Maham and A.~Hjørungnes, ``Minimum power allocation in {SER} constrained
  amplify-and-forward cooperation,'' in {\em Proc. IEEE Vehicular Technology
  Conference (VTC 2008-Spring)}, (Singapore), pp.~2431--2435, May 2008.

\bibitem{che08}
M.~Chen, S.~Serbetli, and A.~Yener, ``Distributed power allocation strategies
  for parallel relay networks,'' {\em IEEE Trans. Wireless Commun.}, vol.~7,
  no. 2, pp.~552--561, Feb. 2008.

\bibitem{hos05}
A.~Host-Madsen and J.~Zhang, ``Capacity bounds and power allocation for
  wireless relay channels,'' {\em IEEE Trans. Inform. Theory}, vol.~51, no. 6,
  pp.~2020--2040, Jun. 2005.

\bibitem{bro04}
D.~R. Brown, ``Energy conserving routing in wireless adhoc networks,'' in {\em
  Proc. Asilomar Conf. Signals, Syst. Computers}, (Monterey, CA, USA), Nov.
  2004.

\bibitem{rez04}
A.~Reznik, S.~R. Kulkarni, and S.~Verdú, ``Degraded {G}aussian multirelay
  channel: {C}apacity and optimal power allocation,'' {\em IEEE Trans. Inform.
  Theory}, vol.~50, no. 12, pp.~3037--3046, Dec. 2004.

\bibitem{han04}
M.~O. Hansa and M.-S. Alouini, ``Optimal power allocation for relayed
  transmissions over {R}ayleigh-fading channels,'' {\em IEEE Trans. Wireless
  Commun.}, vol.~3, no. 6, pp.~1999--2004, Nov. 2004.

\bibitem{doh04}
M.~Dohler, A.~Gkelias, and H.~Aghvami, ``Resource allocation for {FDMA}-based
  regenerative multihop links,'' {\em IEEE Trans. Wireless Commun.}, vol.~3,
  no. 6, pp.~1989--1993, Nov. 2004.

\bibitem{luo05}
J.~Luo, R.~S. Blum, L.~J. Cimini, L.~J. Greenstein, and A.~M. Haimovich,
  ``Link-failure probabilities for practical cooperative relay networks,'' in
  {\em Proc. IEEE 61st Veh. Technol. Conf., Spring}, pp.~1489--1493, May 2005.

\bibitem{lin05}
Z.~Lin and E.~Erkip, ``Relay search algorithms for coded cooperative systems,''
  in {\em Proc. IEEE Global Telecommun. Conf.}, pp.~1314--1319, Nov. 2005.

\bibitem{zhe05}
H.~Zheng, Y.~Zhu, C.~Shen, and X.~Wang, ``On the effectiveness of cooperative
  diversity in ad hoc networks: A {MAC} layer study,'' in {\em IEEE Int. Conf.
  Acoustics, Speech, Signal Processing}, pp.~509--512, Mar. 2005.

\bibitem{ble06b}
A.~Bletsas, A.~Khisti, D.~P. Reed, and A.~Lippman, ``A simple cooperative
  method based on network path selection,'' {\em IEEE Journal on Selected Areas
  in Communications}, vol.~24, no. 3, pp.~659--672, Mar. 2006.

\bibitem{ble07}
A.~Bletsas, H.~Shin, and M.~Win, ``Outage optimality of amplify-and-forward
  opportunistic relaying,'' {\em IEEE Comm. Letters}, vol.~11, pp.~261--263,
  Mar. 2007.

\bibitem{zha06}
Y.~Zhao, R.~Adve, and T.~J. Lim, ``Symbol error rate of selection
  amplify-and-forward relay systems,'' {\em IEEE Comm. Letters}, vol.~10, no.
  11, pp.~757--759, Nov. 2006.

\bibitem{wan03}
Z.~Wang and G.~Giannakis, ``A simple and general parameterization quantifying
  performance in fading channels,'' {\em IEEE Trans. Commun.}, vol.~51, no. 8,
  pp.~1389--1398, Aug. 2003.

\bibitem{sim00}
M.~K. Simon and M.-S. Alouini, {\em Digital {C}ommunication over {F}ading
  {C}hannels: {A} {U}nified {A}pproach to {P}erformance {A}nalysis}.
\newblock New York, USA: Wiley, 2000.

\bibitem{mah10}
B.~Maham and A.~Hjørungnes, ``Orthogonal code design for {MIMO}
  amplify-and-forward cooperative networks,'' in {\em Proc. IEEE Information
  Theory Workshop (ITW'07)}, (Cairo, Egypt), Jan. 2010.

\bibitem{hor85}
R.~Horn and C.~Johnson, {\em Matrix {A}nalysis}.
\newblock Cambridge, UK: Cambridge Academic Press, 1985.

\bibitem{leo94}
A.~Leon-Garcia, {\em Probability and {R}andom {P}rocesses for {E}lectrical
  {E}ngineering}.
\newblock Massachusetts, USA: Addison-Wesley Publishing Company, 1994.

\bibitem{gra96}
I.~S. Gradshteyn and I.~M. Ryzhik, {\em Table of {I}ntegrals, {S}eries, and
  {P}roducts}.
\newblock San Diego, USA: Academic, 1996.

\bibitem{abr72}
M.~Abramowitz and I.~A. Stegun, {\em Handbook of {M}athematical {F}unctions}.
\newblock New York, USA: Dover Publications, 1972.

\bibitem{jaf05}
H.~Jafarkhani, {\em Space-{T}ime {C}oding {T}heory and {P}ractice}.
\newblock Cambridge, UK: Cambridge Academic Press, 2005.

\end{thebibliography}
\end{document}